\documentclass[11pt]{article}

\usepackage{arxiv}
\usepackage[utf8]{inputenc} % allow utf-8 input
\usepackage[T1]{fontenc}    % use 8-bit T1 fonts

\usepackage{url}            % simple URL typesetting
\usepackage{booktabs}       % professional-quality tables
\usepackage{amsfonts}       % blackboard math symbols
\usepackage{nicefrac}       % compact symbols for 1/2, etc.
\usepackage{microtype}      % microtypography
\usepackage{lipsum}
\usepackage{graphicx}
\usepackage{flafter} % floats never appear before their first reference
\graphicspath{ {./images/} }
\usepackage{mathtools} % for \mathclap

\usepackage[utf8]{inputenc}
\usepackage[numbers,sort&compress]{natbib}

\usepackage{dsfont}
\usepackage{amsmath,amssymb,amsfonts}
\allowdisplaybreaks[2] % let long align blocks break across pages
\usepackage{algorithmic}
\usepackage{graphicx}
\usepackage{glossaries}
\usepackage{textcomp}
\usepackage{xcolor}
\colorlet{tablecolor}{black}
\usepackage{bbm}
\usepackage[capitalize]{cleveref}
\usepackage{amsthm}
\usepackage{tikz}
\usetikzlibrary{positioning}
\theoremstyle{definition}

\newtheorem*{assumption*}{Assumption}
\newtheorem{theorem}{Theorem}

\theoremstyle{definition} % Default style for theorems, lemmas, etc.
\newtheorem{lemma}[theorem]{Lemma}

\newtheorem{proposition}[theorem]{Proposition}
\newtheorem{definition}{Definition}
\newtheorem*{remark}{Remark}

\usepackage{subcaption}
\newcommand{\mh}[1]{{\color{black}#1}}

\usepackage{setspace} 
\renewcommand{\baselinestretch}{0.95}

\usepackage{pgfplots}
\pgfplotsset{compat=1.18} 
\pgfplotsset{set layers=standard}
\usepgfplotslibrary{polar}
\usepackage{bbding}
\title{Co-Investment with Payoff-Sharing Mechanism for Cooperative Decision-Making in Network Design Games}

\author{
 Mingjia He \\
  {\abssize Institute for Dynamic Systems and  Control, ETH Zurich, minghe@ethz.ch}\\
   \And
   Andrea Censi  \\
  {\abssize Institute for Dynamic Systems and  Control, ETH Zurich, acensi@ethz.ch}
    \And
  Runyu Zhang \\
  {\abssize  Laboratory for Information and Decision Systems, Massachusetts Institute of Technology, runyuzha@mit.edu} \\
  \And
  Emilio Frazzoli \\
  {\abssize  Institute for Dynamic Systems and  Control, ETH Zurich, emilio.frazzoli@idsc.mavt.ethz.ch} \\
  \And
 Gioele Zardini \\
  {\abssize  Laboratory for Information and Decision
Systems, Massachusetts Institute of Technology, gzardini@mit.edu} 
}

\begin{document}
\maketitle
\begin{abstract}
\fontsize{10}{12}\selectfont % slightly larger than the 9pt arxiv.sty default
Network-based systems are inherently interconnected, with the design and performance of subnetworks being interdependent. However, the decisions of self-interested operators may lead to suboptimal outcomes for users and the system as a whole.  
This paper explores cooperative mechanisms that can simultaneously benefit both operators and users. 
We address this challenge using a game-theoretical framework that integrates both non-cooperative and cooperative game theory. 
In the non-cooperative stage, we propose a network design game in which subnetwork decision-makers strategically design their local infrastructures. 
In the cooperative stage, co-investment with payoff-sharing mechanism is developed to enlarge collective benefits and fairly distribute them, supporting cooperative decision-making within a competitive environment.
To demonstrate the effectiveness of our framework, we conduct case studies on the Sioux Falls network and real-world public transportation networks in Zurich and Winterthur, Switzerland. 
Our evaluation considers impacts on environmental sustainability, social welfare, and economic efficiency. 
The results indicate that small upfront co-investment can lead to substantial long-term system improvements. 
Furthermore, operators’ diversity provides significant potential for performance improvement under the proposed mechanism. 
In the Zurich–Winterthur case, we examine the influence of bargaining power and strategic exploitation, showing that these factors strongly impact both individual benefits and the willingness to cooperate. 
The proposed framework provides a foundation for improving interdependent networked systems by enabling strategic cooperation among self-interested operators.
\end{abstract}

% keywords can be removed
%\keywords{First keyword \and Second keyword \and More}

%\usepackage{glossaries}

% Sets and Spaces
\newcommand{\graph}{\Gamma}
\newcommand{\graphspace}{\boldsymbol{\Gamma}}
\newcommand{\demandmodel}{\mathcal{Y}}
\newcommand{\operatorSet}{\mathcal{I}}
\newcommand{\actionspacei}{\mathcal{H}_i}
\newcommand{\actionspace}{\mathcal{H}}
\newcommand{\nonnegativenumbers}{\mathbb{R}_{\geq 0}}
\newcommand{\nonnegativeintegernumbers}{\mathbb{N}_{\geq 0}}
\newcommand{\nodes}{\mathcal{V}}
\newcommand{\edges}{\mathcal{E}}
\newcommand{\labelsmap}{\ell}
\newcommand{\labels}{\mathcal{L}}
\newcommand{\networkgame}{\mathsf{NG}}

\newcommand{\edgesi}{\mathcal{E}_i}
% Variables
\newcommand{\actioni}{h_i}
\newcommand{\budgeti}{B_i}
\newcommand{\payoffi}{f_i}
% weights
\newcommand{\weighte}{\omega^e_i}
\newcommand{\weightc}{\omega^c_i}
\newcommand{\weightp}{\omega^p_i}

\newcommand{\tup}[1]{\left(#1\right)}
\newcommand{\set}[1]{\left\{ #1 \right\}}

\newacronym{acr:ndg}{NDG}{Network Design Game}
\newacronym{acr:ne}{NE}{Nash equilibrium}
\newacronym{acr:pt}{PT}{public transport}
\newacronym{acr:minlp}{MINLP}{mixed-integer nonlinear program}
\newacronym{acr:ibr}{IBR}{Iterative Best Response}
\newacronym{acr:micp}{MICP}{Mixed-Integer Convex Program}
\newacronym{acr:nbs}{NBS}{Nash Bargaining Solution}
% comment

\newcommand{\m}[1]{\textcolor{orange}{#1}}

\section{Introduction}
\label{sec:introduction}
Globalization has deepened the interconnections among economic, political, and technological systems, amplifying their complexity and mutual dependence~\cite{helbing2013globally}. 
For instance, the World Bank reports that global trade rose from 50\% of gross domestic product (GDP) in 2000 to 63\% in 2022~\cite{globaltrade}, while the number of international migrants reached 281 million in 2020, 183\% higher than in 1990~\cite{WorldMigrationReport}.
These trends, coupled with ongoing population growth, have reshaped key infrastructures such as healthcare, transportation, communication, production, and energy systems.
At the core of these infrastructures are networks, which serve as the backbone for moving people, goods, and information.
They typically consist of multiple subnetworks, each managed by distinct decision-makers with their own objectives, constraints, and incentives.
Due to their inter-connectivity, decisions made in one domain inevitably affect others.
Such interdependence means that \emph{strategic interactions} among network designers are not merely peripheral; they are central to determining both local and system-wide performance.
When decision-makers prioritize their own objectives, the overall system often suffers from suboptimality~\cite{paccagnan2019nash,basar1999dynamic}.
This is particularly evident in the design of transportation infrastructure networks. 
For instance, in cross-border railway projects, a lack of coordination has been identified as a key bottleneck to efficient long-distance travel~\cite{cats2025shift,GROLLE2024103906}.
Similarly, in many regions, urban public transport is run by large operators, while rural services depend on smaller local providers; insufficient integration between the two undermines both rural viability and urban growth~\cite{Zeng2023OptimizationOE}.
Sometimes, however, \emph{strategic cooperation} between operators can yield benefits that far exceed those achievable through isolated action.
A notable example occurred in October 2024, when Switzerland and Germany reached an agreement for Switzerland to invest €50 million in electrifying sections of the German railway network~\cite{Wintle2024}.
At first glance, funding infrastructure abroad, especially when domestic projects remain unfunded, may seem counterintuitive.
Yet, in this case, the electrified segments will shorten travel times between Basel and St. Gallen by roughly 20 minutes~\cite{horpeniakova2024railmarket}, improving rail's competitiveness against road transport.
This case illustrates a broader principle: when networks are interconnected, investing in another operator's infrastructure can deliver greater system-wide gains than investing solely within one's own domain.
However, forging such agreements in multi-agent environments is far from straightforward.
It requires understanding \emph{when} cooperation should occur, \emph{what} joint projects should be pursued, and \emph{how} the resulting benefits should be shared fairly.

In this context, game theory offers a natural framework for analyzing such strategic settings (see, e.g.~\cite{zardini2021game,zardini2023strategic, he25, seo25tcns} for previous work).
It models agents as rational, self-interested decision-makers and distinguishes between \emph{non-cooperative games}, where agents act independently without binding commitments, and \emph{cooperative games}, where agents can form agreements to coordinate actions and share gains. 
In reality, many systems blend these two extremes.
In supply chain management, for example, competing suppliers may share production of transport capacity to reduce costs~\cite{simatupang2002collaborative}. 
In international economics, governments manage domestic policy while negotiating trade agreements~\cite{irwin2024does}.
Likewise, in network design, subnetwork operators may invest independently in their own infrastructure, yet also negotiate cross-network investments or joint projects.
To support such decision-making, we develop a game-theoretic framework for interactive network design (\cref{fig:co_framework}) that explicitly integrates both cooperative and non-cooperative elements.
Central to this framework are two mechanisms: i) co-investment, enabling multiple operators to jointly finance infrastructure projects, and ii) payoff-sharing, ensuring the resulting collective gains are allocated in a fair and incentive-compatible manner.
By embedding these mechanisms into the network design process, we aim to align the interests of self-interested operators while simultaneously improving outcomes for the system's end users.

\begin{figure}[tb]
    \centering\includegraphics[width=0.8\linewidth]{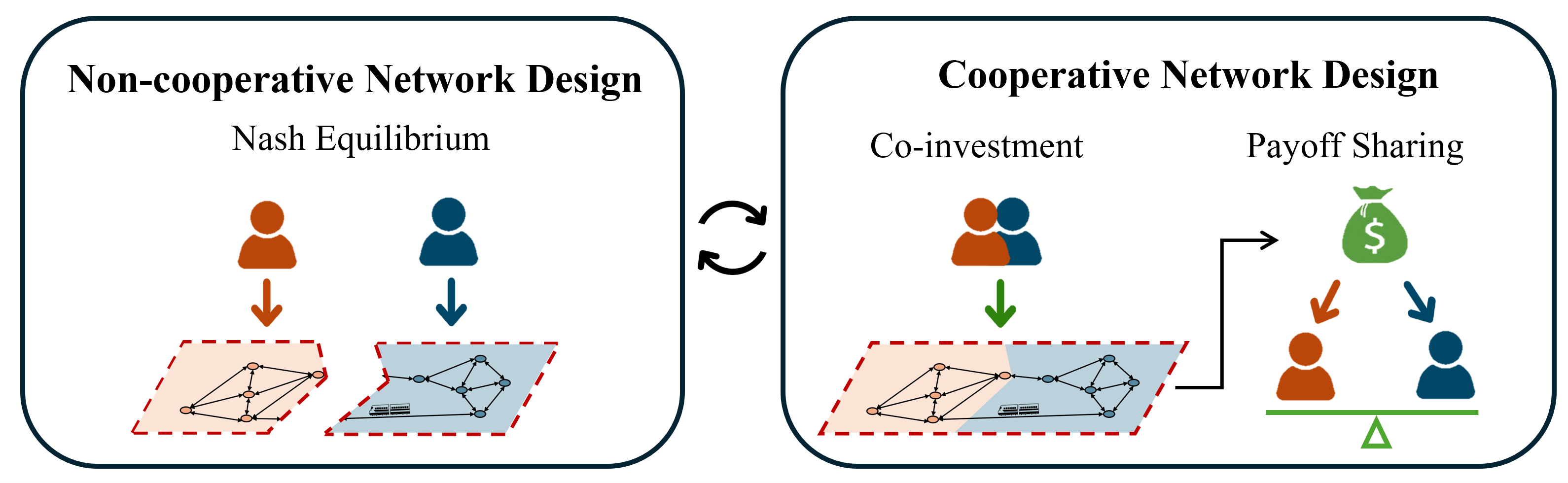}
    \caption{The interactive network design framework, featuring a non-cooperative, as well as a cooperative phase. }
    \label{fig:co_framework}
\end{figure}
\subsection{Related research}
In transportation planning, network design is a crucial strategic decision that significantly impacts the operation of mobility systems~\cite{FARAHANI2013281,liu2021review, Yang01071998,zardini2022analysis}. 
Typical design actions include modifying the network configuration (e.g., adding or removing links) and adjusting link properties (e.g., capacity, frequency, and service quality)~\cite{Luathep2011GlobalOM}.
Mobility systems typically involve multiple stakeholders, including both system operators and users. 
Accordingly, strategic interactions in network design problems can be categorized into two types: network–user interactions and network–network interactions.

\textit{Network-user interactions} lie at the interface between network operators, who set network design parameters, and users, whose travel behavior determines the realized system performance.
Operators often pursue user-oriented objectives such as minimizing total travel time~\cite{gao2005solution}, maximizing service coverage~\cite{lo2009time}, or increasing profit~\cite{dimitriou2008genetic}. 
Users, in turn, select modes and routes based on the available infrastructure, influencing whether operator objectives are achieved.
As a result, most network design models explicitly incorporate user reactions to design decisions.
In this context, a common approach is to formulate the problem as a bi-level optimization problem, where the upper level models the operator's network design decisions, and the lower level captures user route/mode choices on the designed network.
This hierarchical structure enables the network operator to anticipate user responses when optimizing the network.
Lower-level models often include traffic assignment (deterministic or stochastic), discrete choice, and hybrid combinatorial formulations.
For instance, the User Equilibrium (UE) model assumes selfish travelers choosing routes to minimize their own costs, while the Social Optimum (SO) one assumes cooperative travelers minimizing total system travel time~\cite{wardrop1952road}.
Using UE, the design problem becomes a Mathematical Program with Equilibrium Constraints (MPEC), a formulation that is challenging due to nonlinearity and equilibrium conditions.
Notable solution approaches include mixed-integer linear programs (MILPs) with continuous capacity variables~\cite{WANG2010482} and generalized MILP formulations with both continuous and discrete variables~\cite{Luathep2011GlobalOM}.
These works develop algorithms (e.g., cutting constraint methods) to obtain globally or near-globally optimal solutions.
Some work \cite{WANG2015213} addresses non-convex designs combining discrete link additions with continuous capacity expansion, solved to global optimality via linearization.
%Most of these studies assume deterministic user behavior.
To capture randomness in travel choices, researchers have adopted stochastic UE and discrete choice models~\cite{ben1985discrete,Bierlaire2003BIOGEMEAF,Ricard202550YO,meister2023route,HUANG20201}.
Examples include stochastic UE for mixed traffic in autonomy-dedicated facility deployment~\cite{ZHANG2023102784}, and logit-based stochastic UE for modeling EV drivers' routing and recharging decisions in dynamic wireless charging systems~\cite{LIU2021102187}.
This body of work has yielded valuable insights for single-region network design, accounting for network-user coupling.
However, it typically assumes the network is managed by a single entity and does not fully capture interactions among multiple subnetworks.

\textit{Network-network interactions} occur when multiple subnetworks, often under different authorities, are interconnected and mutually influence each other's performance.
Each subnetwork designer may have distinct objectives, budgets, and operational constraints, yet their decisions are interdependent.
Such reasoning applies to various geographical scales, for instance, international networks that connect multiple countries~\cite{itf2023,liu2023global}, and urban-rural networks that involve different municipalities~\cite{porru2020smart}. 
When subnetworks are designed in isolation, inefficiencies are common.
Cross-border railway services, for instance, are often under-prioritized in national investments plans, resulting in lost demand and reduced competitiveness for rail~\cite{medeiros2019cross,cb2022}.
Some studies address these inefficiencies through centralized, integrated design.
For instance,~\cite{GROLLE2024103906} optimize European high-speed rail design and frequency setting from a continental perspective, emphasizing cross-border corridor investments.
Further,~\cite{wang2023integrated} design a multimodal freight network (rail, truck, maritime) to serve heterogeneous shipper needs at minimum total shipping cost, and~\cite{wang2022large} integrate road, transit, and bike subnetworks in a tri-level problem to maximize capacity while accounting for mode split and traffic assignment.
other work models subnetworks as controlled by self-interested stakeholders, leveraging game theory to study competitive dynamics.
Chow and Sayarshad ~\cite{chow2014symbiotic} design coexisting transportation networks (e.g., bike-sharing) via multi-objective optimization, considering mutual impacts.
Further~\cite{lanzetti2023interplay} propose a market model of interactions between autonomous mobility providers, public transport authorities, and customers, showing how ride-hailing influences urban mobility.
Bakhshayesh and Kebriaei~\cite{Bakhshayesh2023} proposed a generalized aggregative game to model the interactions of electric vehicles, a power distribution system operator, charging station operators, and a transportation network operator.  A decentralized learning algorithm is developed to reach the Wardrop equilibrium, and tested on Savannah’s transport model and the IEEE 33-bus network.

In conclusion, for network-network interactions, fully centralized cooperation is often unrealistic due to misaligned local incentives, while purely competitive designs can yield globally suboptimal results.
This has motivated interest in hybrid approaches that blend competition with selective cooperation, enabling joint action where it benefits both operators and users.
In this context, our work proposes a game-theoretic framework for interactive network design that explicitly incorporates both cooperative and non-cooperative elements.
By modeling co-investment opportunities and payoff-sharing mechanisms, our framework aims to identify, negotiate, and fairly allocate the gains from cross-network cooperation, providing actionable decision support in competitive multi-operator environments.

\subsection{Statement of Contribution}
This work makes three main contributions.
First, we introduce a unified game-theoretic framework for the interactive network design problem that explicitly captures both network-user and network-network interactions.
The formulation accommodates multiple decision-makers with distinct objectives, allowing for the systematic analysis of strategic behavior across interconnected subnetworks.
Second, we design cooperative mechanisms for co-investment and payoff-sharing that enable the identification and evaluation of mutually beneficial collaborations.
These mechanisms are structured to balance competitive incentives with cooperative gains, ensuring that improvements benefit both operators and end-users.
Finally, we validate the proposed framework through comprehensive case studies on the Sioux Falls and Zurich-Winterthur networks.
These case studies demonstrate the applicability, efficiency, and practical relevance of our approach, and provide actionable insights into when and how cooperation should take place, as well as how benefits can be fairly distributed while preserving the autonomy of regional stakeholders.
%Overall, the proposed approach advances the understanding of strategic interactions in multi-regional network design and offers decision-support tools to guide coordinated, user-centered improvements in complex, multi-operator systems. 

% Building on our earlier work in~\cite{he25}\footnote{This paper was a Best Student Paper Award finalist at ACC 2025.}{},
% this paper advances the network design game in several key dimensions. 
% Theoretically, we generalize the standard non-cooperative framework and incorporate the cooperative mechanism, formalizing a hybrid game structure.  Within this framework, we provide a detailed model specification and derive its fundamental equilibrium properties, offering rigorous proofs for the existence of the \gls{acr:ne}.
% On the practical side, we enrich the cooperative mechanism to account for two critical real-world factors, bargaining power and the potential for strategic exploitation, that strongly influence the distribution of cooperative gains.
% %
% These improvements allow the model to capture subtler and more realistic dynamics in multi-operator settings.
% The improved framework is further validated through a real-world Zurich-Winterthur case study based on observed network and demand data, illustrating how these new elements can materially affect outcomes and offering decision-support tools for coordinated, user-centered improvements in complex, multi-operator systems.

\mh{Building on our earlier work in~\cite{he25}\footnote{This paper was a Best Student Paper Award finalist at ACC 2025.}, this paper advances the network design game along three dimensions. 
Theoretically, we generalize the noncooperative framework (Definition~\ref{def:NDG}) and formalize the hybrid two-stage structure, and we provide the equilibrium analysis of Section~\ref{sec:theoretical}, including the convexity result (Lemma~\ref{lemma:MICP}), the existence of a pure \gls{acr:ne} (Proposition~\ref{prop:existence}), and the existence and uniqueness of the cooperative solution (Proposition~\ref{prop:payoff}), which were not present in the conference version. 
In addition, we enrich the payoff-sharing scheme with bargaining power (\cref{eq:alpha}) and strategic exploitation (\cref{eq:epsilon_qQ} to \cref{eq:epsilon_vq}), together with the MGR and SET concepts, none of which appeared previously. 
Empirically, we add the practical Zurich and Winterthur case study based on observed network and demand data. The core framework (specifically the two-stage \gls{acr:ndg} with co-investment and payoff sharing, the discrete choice demand model, and the Sioux Falls benchmark) was introduced in~\cite{he25}.}

The remainder of the paper is organized as follows.  
We formalize the network design game and define key concepts in \cref{sec:Game-theoretical Framework}. Building on this foundation, \cref{sec:coinvestment} introduces the co-investment with payoff-sharing mechanism. 
The theoretical properties of this proposed mechanism are examined in \Cref{sec:theoretical}. 
\Cref{sec:exp} presents numerical experiments to evaluate our approach. We discuss the limitations and directions for future work in \cref{sec:limitations}, before concluding in \cref{sec:conclusion}.
\section{Game-theoretical Framework for Network Design Problem} \label{sec:Game-theoretical Framework}
To capture the strategic interactions in the network design setting, 
we first introduce a general game-theoretical framework and define the notion of \gls{acr:ndg}, 
% and then specify mobility systems, including mobility network and travel demand in \Cref{sec:Mobility Network} and \Cref{sec:Travel Demand}, and mobility operators in \Cref{sec:Operators}.
and then specify mobility systems, including networks, demand, and operators.

\subsection{Network Design Game} 
\label{sec:Network Design Game}
Consider a set of self-interested operators (players)~$\operatorSet=\{1,\ldots,N\}$, each controlling a subset of components within a shared network.
The network is represented by an edge-labeled directed graph~$\graph=\tup{\nodes,\edges,\labelsmap}$, where~$\nodes$ is the set of vertices,~$\edges \subseteq \nodes \times \nodes$ is the set of directed edges and~$\labelsmap: \edges \to  \labels $ is a mapping from the set of edges $\edges$ to the set of edge labels $\labels$.
Each operator~$i$ acts on a local subgraph~$\graph_i\subseteq \graph$ (i.e., regions) and aims to maximize a payoff function.
\begin{definition}[Network Design Game] \label{def:NDG}
A \emph{network design game} is defined by the tuple~$\networkgame=\tup{\operatorSet, \graph, \demandmodel, \tup{\actionspacei, \payoffi, \budgeti}_i}$, 
where $\mathcal{I}$ denotes the set of self-interested operators, indexed by $i$, $\graph$ denotes the overall mobility network.
%, and~$\graph_i\subseteq \graph$ represents the subnetwork controllable by operator $i$.
The travel demand model is given by $\demandmodel$.
%Together, the network $\graph$ and the demand model~$\demandmodel$ define the \emph{environment} of the game.
Each operator $i \in \operatorSet$ is characterized by a tuple $\tup{\actionspacei, \payoffi, \budgeti}$, where: 
\begin{itemize}
    \item $\actionspacei:=\{0,1\}^{\vert\edges_i\vert} \times \nonnegativenumbers^{\vert\edges_i\vert}$ is the strategy space of operator $i$, where $\vert\edges_i\vert$ are the number of edges in local network $\graph_i$. These are binary decisions for edge constructions and non-negative continuous decisions for edge capacities. 
    %An action taken by operator $i$ is denoted as $\actioni = \{(x_e, c_e)\}_{e \in \edgesi} \in \actionspacei$, where $x_e$ indicates the construction decision and $c_e$ the capacity assigned to edge $e$, and $\edgesi$ is the subset of edges in the network controlled by operator $i$.  
    \item $\payoffi : \actionspacei \times \mathbb{R}_{\geq 0}^{\vert\edges\vert} \to \nonnegativenumbers$ denotes the payoff function of operator~$i$, which maps the design strategy $h_i \in \actionspacei$ and a vector of non-negative edge flows $y \in \mathbb{R}_{\geq 0}^{|\edges|}$ to a non-negative payoff value. 
    %$|\edges|$ is the total edge number in the network $\graph$.
    \item $\budgeti \in \nonnegativenumbers$ denotes the total budget of operator~$i$ for infrastructure development.
\end{itemize}
Given the strategies~$h_{-i}$ of all other operators, each operator $i \in \operatorSet$ solves the following optimization problem:
\begin{subequations}\label{eq:ndg}
    \begin{align}
    \max_{h_i \in \actionspacei} \quad & f_i\tup{h_i,y}\\
    \text{s.t.} \quad
     &y=\demandmodel \tup{h_i,h_{-i},\graph}\\
     &b_i \tup{h_i} \leq \budgeti,
\end{align}
\end{subequations}
where the function $b_i: \actionspacei \rightarrow \nonnegativenumbers$ maps a specific strategy to its implementation cost.
And $\demandmodel: \actionspace \times \graphspace \rightarrow \nonnegativeintegernumbers^{|\edges|}$ maps from a network design strategy profile and a graph to the vector of edge flows,
where $\graphspace$ denotes the space of network graphs. 
The vector $y=\tup{y_e}_{e \in \edges} \in \nonnegativeintegernumbers^{|\edges|}$ represents the served flow on all edges.

\end{definition}

A key solution concept in game theory is the \gls{acr:ne}~\cite{nash1950equilibrium}, leveraged to study interactions among rational players.
For a \gls{acr:ndg}, the network at \gls{acr:ne} is a stable outcome where no operator can improve their payoff by unilaterally deviating from their network design strategies.

\begin{definition}[$\text{Nash Equilibrium}$ of \gls{acr:ndg}] \label{def:ne}
A strategy profile $\tup{h_i,h_{-i}}$ is a \emph{Nash Equilibrium} of the \gls{acr:ndg} if, for every operator $i \in \operatorSet$, 
$f_i\tup{h_i,y} \geq f_i \tup{h'_i,y'}, \ \forall h'_i \in \actionspacei$,
where $y=\demandmodel \tup{h_i,h_{-i},\graph}$, and $y'=\demandmodel \tup{h'_i,h_{-i},\graph}$.
%f_i(h_i,h_{-i}) \geq f_i(\hat{h}_i,h_{-i},y), \ \forall$ $\hat{h}_i \in \actionspacei$.
\end{definition}

This concept serves as a reference point for identifying inefficiencies and for designing cooperative mechanisms.

\begin{remark} [Generality of the framework]
The framework captures both network–network interactions (among operators) and network–user interactions (users responding to design through the demand model~$\demandmodel$).
It can generalize across infrastructure domains (e.g., transportation, energy, communications) involving stakeholders with distinct objectives and budget constraints.
Operators may control geographically adjacent subnetworks (multi-region systems~\cite{cats2025shift}) or overlapping layers (multimodal systems~\cite{Bakhshayesh2023}).
We assume perfect information, so all operators observe each other's strategies, reasonable for public infrastructure.
\end{remark}

The resulting equilibrium properties depend on the specific functional forms of the demand model and the operators’ objectives. To demonstrate the adaptability of the proposed framework, we present a tractable instantiation comprising: i) a two-region graph, ii) a discrete-choice-based travel demand model~\cite{ben1985discrete, Bierlaire2003BIOGEMEAF}, and iii) operator-specific strategies and payoff functions.

\subsection{Mobility Systems} 
\label{sec:Mobility Systems}
\subsubsection{Mobility Network} 
\label{sec:Mobility Network}
%We model the mobility network as an edge-labeled directed graph~$\graph= \tup{\nodes, \edges, \labelsmap}$.
%\cathy{Let's try making the definitions more organized. The logic is that we should let readers not interested in details to also understand what we roughly talk about and relates to the high-level network design game. We should follow roughly the flow of i) A one-sentence overview, for example, for mobility network subsection, we should add one summarizing sentence as "we address the detailed modeling of mobility network $\Gamma$" and for the travel demand, it addresses how we get $y=\mathcal{Y}...$.  ii) catagorize the heavy notations to different paragraphs, such as: "Labels $\ell(e)$ for edge $e$", "partition of edges", "multimodal mobility choices"}
% Recall the mobility network~$\graph=\tup{\nodes,\edges, \labelsmap}$.

We address the detailed modeling of the mobility network $\graph=\tup{\nodes,\edges, \labelsmap}$.
For each edge~$e$, we assign a label~$\ell(e)=\tup{x_e,c_e,l_e,t_e}\in \labels=\set{0,1} \times \tup{\nonnegativenumbers \cup \{\infty\}}\times \nonnegativenumbers \times \nonnegativenumbers$, characterized by the availability of the mobility service on edge $x_e$, the capacity on the edge $c_e$, the edge length $l_e$, and the travel time associated to the edge~$t_e$.
\begin{figure}[tb]
    \centering
    \includegraphics[width=0.8\linewidth]{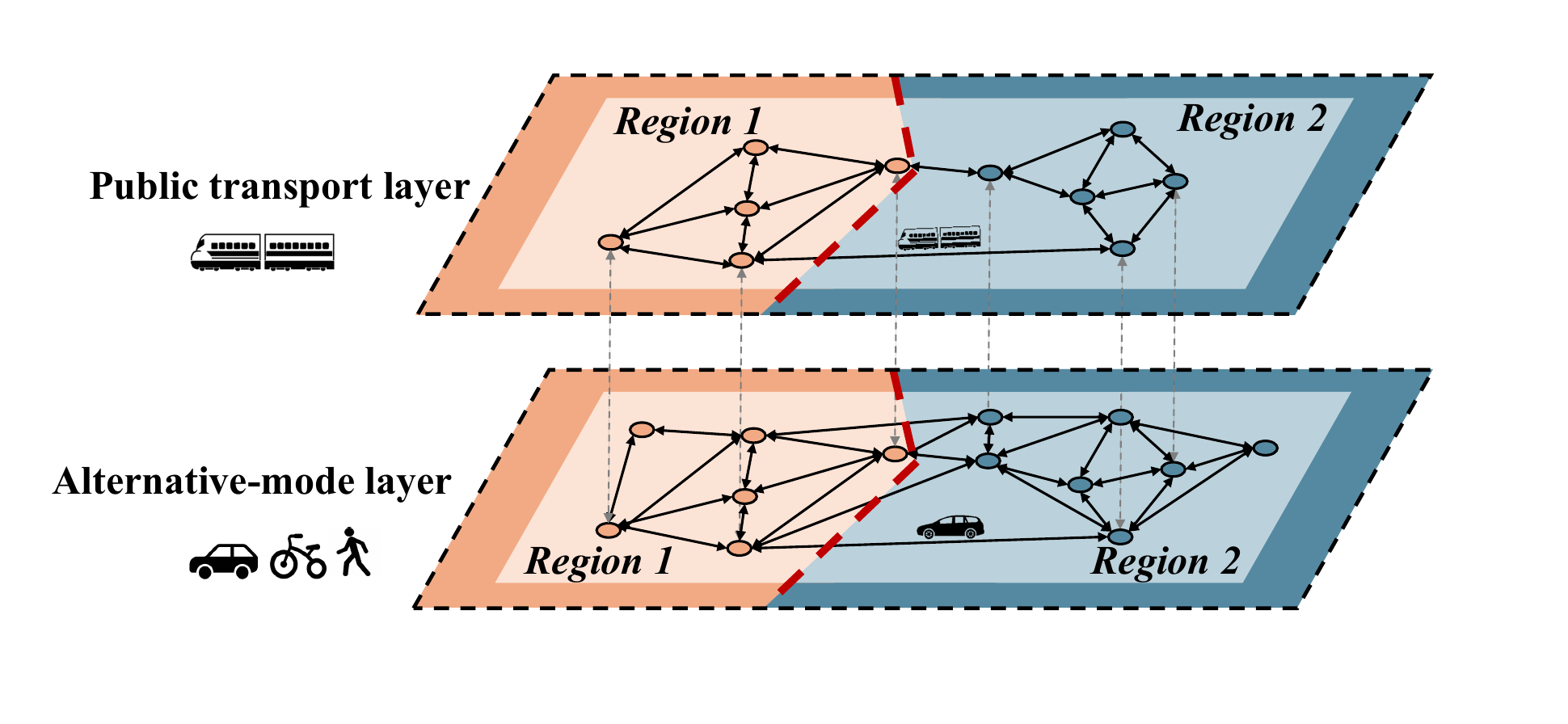}
    \caption{Multimodal mobility network for Region 1 and Region 2.}
    \label{fig:graph}
\end{figure}

\paragraph{Region Partition}
We assume that there are two regional operators $i \in \operatorSet= \{1,2\}$, as shown in \Cref{fig:graph}.
The graph~$\graph$ can be divided into two subgraphs $\graph^1 = \tup{\nodes^1, \edges^1, \labelsmap^1}$ and $\graph^2 = \tup{\nodes^2, \edges^2, \labelsmap^2}$ corresponding to two regions (denoted Region 1 and Region 2 for simplicity), where $\nodes^1$ and $\nodes^2$ are disjoint subsets of $\nodes$ satisfying $\nodes= \nodes^1 \cup \nodes^2$ and $\nodes^1 \cap  \nodes^2 = \emptyset$. The sets of edges for the subgraphs are defined as follows. 
The edge set of Region $i$ is $\edges^i=\{\tup{u,v} \in \edges|u,v \in \nodes^i\}$, and the region-connecting edge set is defined as 
\begin{equation*}
    \edges^c=\{\tup{u,v} \in \edges| u \in \nodes^i, v \in \nodes^j, i,j \in \operatorSet, i \neq j \}.
\end{equation*}
The edge sets satisfy $\edges= \edges^1 \cup  \edges^2  \cup \edges^c $ and $\edges^1 \cap  \edges^2 \cap \edges^c = \emptyset$. 
%\cathy{quick question, in this formulation, we're not optimizing/co-investing over $\edges^c$, right?}
This partition allows each regional network to be designed by regional operators while maintaining the overall connectivity of the mobility network.
\paragraph{Multimodal Mobility}
To enable multimodal mobility choices, each regional subgraph $\tup{\graph^i}_{i \in \{1,2\}}$ contains a \gls{acr:pt} network layer $\graph^i_{P}= \tup{\nodes^i_{P}, \edges^i_{P}, \labelsmap^i_{P}}$ and an alternative-mode network layer $\graph^i_{A}= \tup{\nodes^i_{A}, \edges^i_{A}, \labelsmap^i_{A}}$, which we assume represents an aggregated layer for other transportation modes such as private vehicles, bikes and walking, where~$\nodes^i_{P} \cap \nodes^i_{A} = \emptyset$. 
\gls{acr:pt} networks are characterized by stations $u \in \nodes^i_{P}$ and line segments $\tup{u,v} \in \edges^i_{P}$; the network for the alternative-mode layer is modeled by intersections $u \in \nodes^i_{A}$ and link segments $\tup{u,v} \in \edges^i_{A}$.
The mode-transfer edges set is represented as $\edges^i_{C} \subseteq \nodes^i_{P} \times \nodes^i_{A} \cup \nodes^i_{A} \times \nodes^i_{P}$, allowing the switch of transportation mode during a single trip. 
Similarly, the region-crossing edge set consists of three subsets of edges: \gls{acr:pt} edges, alternative-mode edges and mode-transfer edges, i.e., $ \edges^c = \edges^c_{P} \cup \edges^c_{A} \cup  \edges^c_{C}$.
Given the above definitions, it holds that  $\nodes =  \nodes^1_{P} \cup    \nodes^1_{A} \cup  \nodes^2_{P} \cup  \nodes^2_{A} $and $ \edges =  \edges^1_{P} \cup   \edges^1_{A}\cup   \edges^1_{C} \cup  \edges^2_{P} \cup   \edges^2_{A} \cup   \edges^2_{C}\cup  \edges^c_{P} \cup   \edges^c_{A} \cup   \edges^c_{C}$. By defining subgraphs for regions and layers for \gls{acr:pt} and alternative modes, the framework supports modeling multiregion, multimodal transportation networks.
Local operators manage their regional networks, while users can travel across regions and use multiple transportation modes.

\subsubsection{Travel Demand} \label{sec:Travel Demand}
We address the modeling of traveler decision-making to derive the specification of the demand model $\demandmodel$.
\paragraph{Choice Modeling}
Let $\mathcal{M}$ denote the set of travel requests.
Each request $m \in \mathcal{M}$ is defined as $r_m=\tup{o_m,d_m,\alpha_m} \in \mathcal{R} = \nodes_{A} \times \nodes_{A} \times \mathbb{N}_0$, where $o_m$ and $d_m$ denote the origin and the destination, respectively, and $\alpha_m$ denotes the number of trips associated with the request.
% \begin{assumption} \label{ass:demand}
% The following assumptions apply to traveler choice:
% \begin{enumerate}
%     \item Travelers choose between mode-based route options: a \gls{acr:pt}-prioritized route and a route based on alternative transportation modes (referred to as the \emph{alternative} route hereafter).
%     \item Travelers evaluate routes using a linear utility function, and choice probabilities follow a Logit model;
%     \item The \gls{acr:pt} network has strict capacity limits, with excess demand switching over to the alternative layer.
% \end{enumerate}
% \end{assumption}

% \cathy{Mingjia, could you add some references to support the assumptions and modeling?}
For each travel request, two route options are available: a \gls{acr:pt}-prioritized route and a route based on alternative transportation modes (referred to as the \emph{alternative} route hereafter).
The \gls{acr:pt}-prioritized route prioritizes the use of \gls{acr:pt} services, with alternative modes being used only when \gls{acr:pt} is not available. 
In contrast, the alternative route relies exclusively on other transportation modes.
For the request $r_m$, we assume that a proportion $p_m \in [0, 1]$ of trips will choose the \gls{acr:pt}-prioritized route, which is determined by: 
\begin{align}
&p_{m} = \frac{e^{u^P_m}}{e^{u^A_m} + e^{u^{P}_m}}, \quad  \forall m \in \mathcal{M},\label{eq:p_m}
\end{align}
where  $u^{P}_m$ and $u^A_m$ denote the travel utilities of the \gls{acr:pt}-prioritized route and the alternative route, respectively.

\paragraph{Utility Function}
Travelers evaluate the utilities of both routes based on the travel time and service price:
\begin{subequations}\label{eq:u}
\begin{align}
u^A_m = & - \sum_{e \in \edges^A_m} l_e 
    \tup{
        \frac{\gamma_\mathrm{vot}}{v_A} + \gamma^\mathrm{A}_{1} 
    }, \quad \forall m \in \mathcal{M} \label{eq:uA}, \\
u^P_m = & - 
     \sum_{e \in \edges^P_m} 
     \Bigl[
     x_e l_e 
    \tup{
        \frac{\gamma_\mathrm{vot}}{v_P} + \gamma_{1}^\mathrm{P}
    } \nonumber\\
    & - \tup{1-x_e} \sum\limits_{a \in \edges^A_e}  l_a
    \tup{
        \frac{\gamma_\mathrm{vot}}{v_A} + \gamma_{1}^\mathrm{A}
    }
    \Bigr] , \quad \forall m \in \mathcal{M} \label{eq:uR} ,
\end{align} 
\end{subequations}
where the edge sets for such routes are denoted as $\edges^P_m$ and $\edges^A_m$, respectively. 
$l_e$ is the travel distance on edge $e$, and $\gamma_\mathrm{vot}$ is the value of time.
In \cref{eq:uR}, the utility is determined by the availability of PT service, represented by $x_e$. When the service is available ($x_e = 1$), the first term calculates the utility associated with \gls{acr:pt} travel.
When the \gls{acr:pt} service is unavailable ($x_e = 0$), travelers instead switch to alternative edges.
In this case, $\edges^A_e$ represents the set of alternative-mode edges used as a substitute for \gls{acr:pt} edge $e$.
%\cathy{Mingjia, remember to find a better notation for ${\edges^a_e}$}
The distance-based prices and average speeds are 
$(\gamma_{1}^\mathrm{P}, v_P)$ for \gls{acr:pt} service and $(\gamma_{1}^\mathrm{A}, v_A)$ for the alternative mode.
For the scope of this study, the value of time, service prices, and speeds ($\gamma_\mathrm{vot}, \gamma_{1}^A, \gamma_{1}^P, v_A, v_P$) are assumed to be constant.
%\cathy{We assume ($\gamma_\mathrm{vot}, \gamma_A, \gamma_P, v_A, v_p$) to be constant, right? We should specify that.}

\paragraph{Capacity-Constrained Demand Model}
We assume that the \gls{acr:pt} edge flow depends on both the potential demand $\tilde{y}_e$, defined as the total flow intending to use edge $e$ based on route choices, and the edge capacity $c_e$, which imposes an upper bound on the edge flow.
The potential \gls{acr:pt} demand $\tilde{y}_e$ can be calculated by:
$$\tilde{y}_e = \sum\limits_{m \in \mathcal{M}}  \mathds{1}_{e \in \edges^P_m} \alpha_m p_m, \quad \forall e \in \edges^P,$$
where $\mathds{1}_{e \in \edges^P_m}$ equals 1 if edge $e$ belongs to the \gls{acr:pt}-prioritized route for request $m$, and 0 otherwise.

The realized \gls{acr:pt} flow $y_e$ is subsequently affected by the edge capacity $c_e$. For \gls{acr:pt} edges, the flow is capped at $c_e$, forcing any excess demand to switch modes. 
For alternative edges, the realized flow consists of two components: travelers who initially choose the alternative mode and demand diverted from \gls{acr:pt} services due to capacity constraints.
\begin{equation}
y_e =\!
\begin{cases}
    \min(\tilde{y}_e, c_e), & \text{if } e \in \edges_P, \\[8pt]
    \displaystyle \!\!\sum_{m \in\! \mathcal{M}}\!\! \mathds{1}_{e \in \edges^A_m} \alpha_m (1\!-\!p_m) \!+ \!\!\sum_{a \in \edges^P_e} (\tilde{y}_a \!-\! c_a)^+, & \text{otherwise},
\end{cases}
\label{eq:y}
\end{equation}
\noindent where 
$\mathds{1}_{e \in \edges^A_m}$ indicate whether edge $e$ is within the alternative route. $\edges^P_e$ denotes the set of \gls{acr:pt} edges for which alternative edge $e$ serves as a substitute. The operator $(\tilde{y}_a - c_a)^+$ quantifies the capacity spillover from these edges to alternative edge $e$.

\mh{
We adopt the user decision-making model in~\cref{eq:y} for the subsequent analysis.
Based on this model choice, we can derive conditions under which the individual operator problem is a mixed-integer convex program, thereby enabling the global-optimality guarantees and the existence of \gls{acr:ne} for the \gls{acr:ndg} established in~\cref{sec:theoretical}.}
An alternative formulation incorporating congestion effects is provided in Appendix~\ref{appendix:traffic assignment}.
Under that formulation, the lower-level equilibrium constraints introduce non-convexities, turning the problem into a multi-leader--multi-follower Stackelberg game (a mathematical programming with equilibrium constraints) for which such guarantees no longer hold; further discussion can be found in our related work~\cite{he2025hierarchical}.

\subsubsection{Self-interested Operators} \label{sec:Operators}
%\cathy{Again, add one sentence overview. And add paragraphs "Strategy $h_i$of Operator $i$" and "Utility of Operator $i$".} \cathy{Should be more specific on the action, we are only modifying the labels $\ell(e)$ where $e\in \edges_P$. }

Regional operators make decisions independently for their respective networks and may adopt different objectives and performance evaluation criteria. 
We therefore introduce the specification of the tuple $\tup{\actionspacei, \payoffi, \budgeti}$ for each operator $i \in \operatorSet$.
Their decision-making is focused on outcomes within their own regions, rather than the impacts on the broader system.

\paragraph{Decision Variables}
A strategy of by operator $i$ is denoted by $h_i := \{\tup{d_e, s_e}\}_{e \in \edges_i} \in \mathcal{H}_i$, where $\edges_i \subseteq \edges$ represents the subset of the network edges under the control of operator~$i$. 
For each edge $e \in \edges_i$, the binary variable $d_e \in \{0,1\}$ represents the construction decision, where $d_e = 1$ indicates that edge $e$ is chosen to be constructed.
%, and $d_e = 0$ indicates that it is not constructed.  
The variable $s_e \in [0, s_{\max}]$ denotes the service frequency assigned to edge $e$, , bounded by the maximum allowable frequency $s_{\max}$.

These decisions directly modify the edge labels of the graph $\graph$. 
Recall from \Cref{sec:Mobility Network} that each edge is associated with a label $z_e = (x_e, c_e, l_e, t_e)$. The strategy $h_i$ updates the first two components: the construction status $x_e$ and the capacity $c_e$.
Let $\graph_{\text{in}}$ denote the input network configuration with edge availability and capacity given by $\{\hat{x}, \hat{c}\}$.
We define the network state transition $\mathcal{T}: (\graph_{\text{in}}, h_i) \to \graph_{\text{out}}$, where the post-game network $\graph$ with edge availability and capacity $\{x, c\}$, is determined by:
\begin{subequations} \label{eq:state_transition}
\begin{align}
    x_e &= \hat{x}_e + d_e, \label{eq:trans_x}\\
    c_e &= \hat{c}_e + \kappa s_e. \label{eq:trans_c}\\
    x_e &\leq s_e \leq x_e \Omega, \label{eq:trans_consistency}
\end{align}
\end{subequations}
\noindent Eq.~\eqref{eq:trans_x} updates the topology. Eq.~\eqref{eq:trans_c} establishes that the effective capacity $c_e$ scales linearly with the service frequency $s_e$ via the coefficient $\kappa$.
Constraint~\eqref{eq:trans_consistency} enforces logical consistency using the Big-M method (where $\Omega$ is a large positive constant), ensuring that positive service frequency $s_e$ can only be assigned if the edge is active ($x_e=1$).
% We explicitly separate the construction variable $d_e$ from the capacity variable $s_e$. 
% This distinction allows us to decouple infrastructure investment (driven by $d_e$) from operational costs (driven by $s_e$).
%
%

\paragraph{Performance Metrics}
The performance of mobility networks can be evaluated from multiple perspectives. 
We assume that regional operators will consider the environmental, social, and economic impacts. 
Specifically, operator $i$ quantifies CO\textsubscript{2} emissions, total travel costs, and profitability generated within its own region, denoted by $J^{e}_i$, $J^{c}_i$, and $J^{p}_i$, respectively. 
These performance metrics depend not only on the network design of their region but also on the network of the other one.  
This interdependence is captured by edge flows $y = \demandmodel \tup{h_i, h_{-i}, \graph}$, where $h_i$ denotes the network design strategy of regional operator~$i$, and $h_{-i}$ represents the strategy of the other region.
Thus, travelers' choices are influenced by overall network strategies and the existing network layout (see \cref{eq:y}).
The payoff function for operator~$i$ is then given by:
\begin{align}
    & f_i\tup{h_i, y}=
    - \omega^e_i J^e_i \tup{y} 
    - \omega^c_i J^{c}_i\tup{y} 
    + \omega^p_i J^p_i\tup{h_i, y}, \label{eq:f}
\end{align}
where $\weighte, \weightc, \weightp \in \nonnegativenumbers^3$ are weights reflecting the relative importance that operator~$i$ assigns to environmental impact, travel cost, and profitability, respectively.
Performance metrics can be determined by:
\begin{align}
    &J^{e}_i\tup{y}= 
        \sum_{e \in \edges^i_{P}} \gamma_{2}^\mathrm{P} l_e y_e
        + \sum_{e \in \edges^i_A} \gamma_{2}^\mathrm{A} l_e y_e, \nonumber\\
    &J^{c}_i\tup{y}=  
        \sum_{\mathclap{e \in \edges^i_{P}}} l_e y_e \tup{\frac{\gamma_\mathrm{vot}}{v_P}\!+\!\gamma_{1}^\mathrm{P}}
         \!+\!\sum_{\mathclap{e \in \edges^i_A}} l_e y_e \tup{\frac{\gamma_\mathrm{vot}}{v_A}\!+\!\gamma_{1}^\mathrm{A}}, \nonumber\\
    &J^p_i\tup{h_i, y} = 
        \sum_{e \in \edges^i_{P}}  \gamma_{1}^\mathrm{P} l_e y_e 
        - b_i(h_i),\nonumber
\end{align}
System emission $J^e_i$ accounts for both the \gls{acr:pt} service and alternative services, positively related to the volume of travel demand and distance traveled.
The parameters $\gamma_{2}^\mathrm{P}$ and $\gamma_{2}^\mathrm{A}$ denote the CO\textsubscript{2} emission unit for \gls{acr:pt} and alternative services, respectively. 
The total travel cost $J^{c}_i$ is the travel cost generated from the requests within the region $i$.
Profitability $J^p_i$ of the local network designer is the gap between the revenue from the \gls{acr:pt} service and the construction cost. 
Revenue is calculated from the flow over all local edges, which considers the service price, the length of the edge, and the flow. 

Construction cost $b_i(h_i)$ includes both the base costs and the costs associated with upgrading the service capacity, which is determined by:
\begin{align}
    b_i(h_i) = \sum_{e \in \edges^i_{P}} c^{b} l_e d_e +c^{k} l_e s_e
    \label{eq:cost}
\end{align}
The parameters $c^{b}$,~$c^{k}$ are the unit costs of line construction and capacity enhancement, respectively.

% \cathy{We should add another subsection that describes the NDG rigorously, given that we have all the detailed notations and definitions}

\subsection{Model Instantiation}
By instantiating the general game-theoretic framework in Definition~\ref{def:NDG} with the mobility and operator models from Section~\ref{sec:Mobility Systems}, we obtain the explicit optimization problem for the \gls{acr:ndg}.
In the non-cooperative setting, operators plan simultaneously to maximize individual payoffs. Each operator $i \in \operatorSet$ acts independently, treating others’ strategies $h_{-i}$ as fixed.
The resulting strategic local optimization problem~$\mathsf{Loc}_i$ for operator~$i$ is formulated as:
\begin{subequations} \label{problem:noncoop}
\begin{alignat}{2}
\max_{h_i \in \mathcal{H}_i} \quad 
    & f_i\bigl(h_i, y\bigr) 
    & \qquad & \text{\footnotesize (see \cref{eq:f})} \label{eq:opt_obj}\\
\text{s.t.} \quad 
    & y = \demandmodel\bigl(h_i, h_{-i}, \graph\bigr) 
    & & \text{\footnotesize (see \cref{eq:y}, (\ref{eq:state_transition}))} \label{eq:opt_demand}\\
    & b_i(h_i) \leq \budgeti 
    & & \text{\footnotesize (see \cref{eq:cost})} \label{eq:opt_budget}\\
    & h_i = \{(d_e, s_e)\}_{e \in \edges_P^i}  \\
    & d_e \in \{0,1\}, \ s_e \in [0, s_{\max}] 
    & & \forall e \in \edges^i_P, \label{eq:opt_vars}    
    %& \text{Eqs.}\tup{\ref{eq:nd_x})-(\ref{eq:nd_cs}} \nonumber,
\end{alignat}
\end{subequations}
\noindent where the objective function \eqref{eq:opt_obj} integrates regional environmental, social, and economic goals; 
constraint \eqref{eq:opt_demand} represents the mapping of network design to edge flows; 
and constraint \eqref{eq:opt_budget} enforces the financial budget for construction and capacity expansion.

\begin{remark}[Solving the strategic local optimization problem]It is important to note that the formulation in \eqref{problem:noncoop} is not a classical isolated optimization problem due to the strategic interaction captured in constraint \eqref{eq:opt_demand}, where the demand flows $y$ depend on the strategies of other agents ($h_{-i}$).
The collection of these problems across all $i \in \operatorSet$ constitutes the \gls{acr:ndg}. 
However, to find the equilibrium network, we can employ the Iterative Best Response algorithm. Within each iteration of this algorithm, we fix the strategies of other operators $h_{-i}$; under this assumption, the problem reduces to an \gls{acr:minlp}.
\end{remark}
\begin{remark}[Problem complexity]
The resulting subproblem for operator $i$ is an \gls{acr:minlp} featuring binary decision variables $\set{d_e}_{e \in \edges_i}$ for edge existence and continuous variables $\set{s_e}_{e \in \edges_i}$ for capacity. 
For each operator $i$, the optimization problem scales with $\mathcal{O}\tup{|\edges_P^i|}$ decision variables and $\mathcal{O}\tup{|\edges_P^i|+|\mathcal{M}'^i|}$ constraints, where $|\edges_P^i|$ denotes the number of \gls{acr:pt} edges in the region $i$. 
$|\mathcal{M}'^i|$ is the set of relevant travel requests, defined as requests originating or terminating in region $i$:
    \begin{equation*}
        |\mathcal{M}'^i|= \sum_{r_m \in \mathcal{R}} \alpha_m \mathds{1}\{o_m \in \nodes^i \vee d_m \in \nodes^i \}.
    \end{equation*}
\end{remark}

\section{Co-investment with Payoff Sharing}\label{sec:coinvestment}
\begin{figure}[tb]
    \centering
    \includegraphics[width=0.8\linewidth]{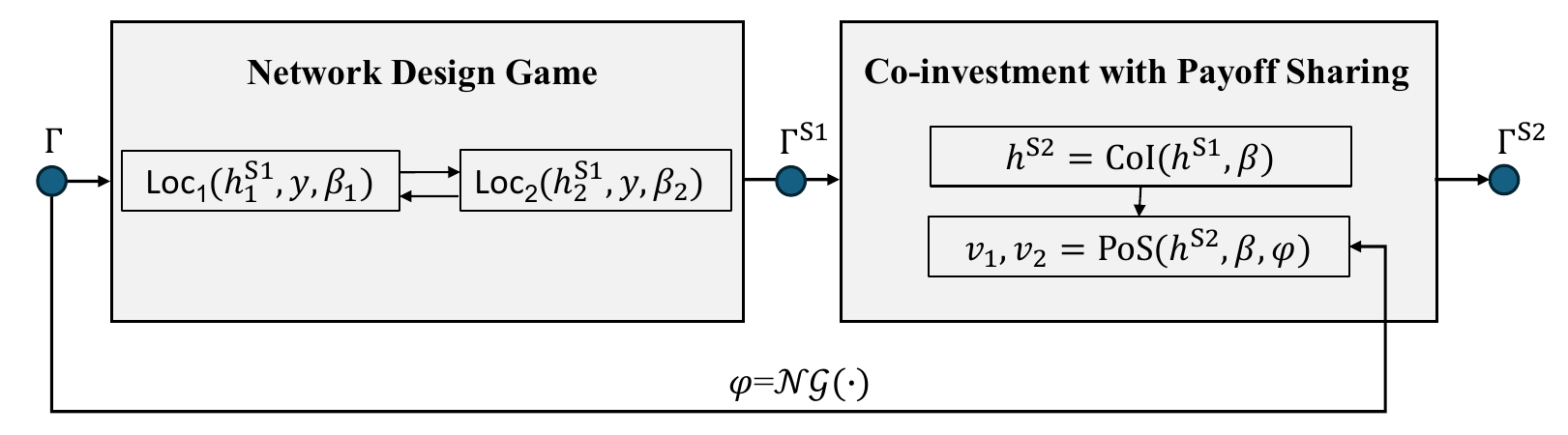}
    \caption{The proposed network design approach incorporates a cooperative network design stage involving co-investment and payoff sharing, which follows the non-cooperative \gls{acr:ndg}. Here, operators can decide whether to invest jointly in the network, determine their individual contributions, and agree on the resulting payoff allocation.}
    \label{fig:model}
\end{figure}
Traditional network design approaches often adopt one of two extremes: a \emph{purely competitive} setting, where operators act independently without coordination, or a \emph{fully cooperative} setting, where they commit to joint planning from the beginning.
In practice, the most promising opportunities for improving overall system performance often lie between these extremes.
We therefore assume that cooperation is voluntary and arises only when it is individually rational, i.e., when operators attain payoffs at least as large as those obtained by acting independently. 
Acting independently corresponds to allocating the full available budget to a non-cooperative \gls{acr:ndg}. 
%Let this benchmark game be denoted by $\networkgame^{ne}$, and let $\varphi_i$ denote the payoff for operator $i$ at the Nash equilibrium off a full-budget non-cooperative \gls{acr:ndg}.

We then propose a hybrid framework that blends competition and cooperation through a two-stage process.
At its core is a co-investment with payoff-sharing mechanism designed to encourage \emph{strategic collaboration} among self-interested operators, enabling joint ventures that can improve outcomes for both operators and users.

\begin{definition}[Two-Stage Coopetitive NDG] \label{def:NDG_coop}
The \emph{Coopetitive Network Design Game} is a sequential game framework that integrates non-cooperative competition and cooperative investment. Operators split their total budget $B_i$ into two components, where $\beta_i \in [0,1]$ denotes the fraction reserved for cooperation. The game proceeds in two stages:

\begin{enumerate}
    \item \textbf{Stage 1 (Non-Cooperative Design):} Operators are involved in a non-cooperative NDG (in Equations~\eqref{eq:ndg}) subject to a restricted individual budget $(1-\beta_i)B_i$. This results in an equilibrium graph $\graph^{S1}$. 
    %\cathy{Refer to eq(1) to remind the reader}

    \item \textbf{Stage 2 (Co-investment and Payoff Sharing):} Operators pool their remaining resources, $\sum_{i \in \mathcal{I}} \beta_i B_i$, to jointly optimize the network and redistribute the resulting gains. 

    \textit{Co-investment:} Operators collectively determine a joint strategy profile $h^{\mathrm{S_2}}$ to maximize the aggregate payoff, conditional on the equilibrium graph from Stage 1:
        \begin{subequations}
        \begin{align}
        \max_{h^{\mathrm{S_2}}\in \actionspace} \quad & \sum_{i \in \mathcal{I}} f_i \left(h^{\mathrm{S_2}}_i, y\right)\\
        \text{s.t.} \quad
         &y=\demandmodel \tup{h^{\mathrm{S_2}},\graph^{\mathrm{S_1}}}\\
         &\sum_{i \in \mathcal{I}} b_i \tup{h^{\mathrm{S_2}}_i} \leq \sum_{i \in \mathcal{I}} \beta_i B_i,
        \end{align}
        \end{subequations}

    \textit{Payoff Sharing:} The total cooperative utility is distributed according to a \gls{acr:nbs}, ensuring that no operator is worse off than in the pure non-cooperative outcome. %\cathy{The writing here can be confusing to first-time readers because $\varphi_i, S$ is not formally introduced. Will come back to fix this, probably going to move this to III.B.2}
\end{enumerate}
\end{definition}

%
%
%\cathy{I really like this added definition that makes things clearer! We can also add a remark to further justify our formulation and compared it with the NDG in II.A}
%\Cref{fig:model} summarizes the proposed approach.
% We introduce the co-investment ratio~$\beta_i\in [0,1]$, which allows each operator~$i$ to divide its budget into two components: a portion~$B_i\tup{1-\beta_i}$ reserved for the initial non-cooperative game, and a portion~$B_i\beta_i$ allocated to the cooperative investment stage.

% Starting from an initial network $\graph$, Stage 1 models a non-cooperative \gls{acr:ndg} based on~$\graph$, where operators solve interdependent local optimization problems ($\mathsf{Loc}_i$) subject to their reduced budgets~$B_i \tup{1-\beta_i}$ and given the design strategies of the other operators (see~\Cref{sec:Design Stage 1}).
% }
% This stage yields a \gls{acr:ne} strategy profile and the corresponding Stage 1 network layout $\graph^\mathrm{S1}$.
% In Stage 2, operators jointly design network expansions through a cooperative investment model~$\mathsf{Col}$.
% The resulting benefits are then allocated according to a payoff-sharing mechanism~$\mathsf{PoS}$ that takes into account the outcomes of the non-cooperative baseline~$\networkgame\tup{\cdot}$ (see \Cref{sec:Design Stage 2}).
\Cref{fig:model} summarizes the proposed approach.
Starting from an initial network $\graph$, Stage 1 models a non-cooperative \gls{acr:ndg}, yielding a \gls{acr:ne} strategy profile $h^\mathrm{S1}$ and the Stage 1 network layout $\graph^\mathrm{S1}$.
In Stage 2, operators jointly design network expansions through a cooperative investment model~$\mathsf{Col}$, and the resulting benefits are allocated via the payoff-sharing mechanism~$\mathsf{PoS}$. The result is the final network configuration~$\graph^\mathrm{S2}$.

\begin{remark}[Comparison with non-cooperative NDG] \label{rem:comparison}
The Coopetitive NDG generalizes the standard non-cooperative framework in Section~\ref{sec:Network Design Game}. 
If the co-investment budget fraction is set to $\beta_i = 0$ for all $i$, the game reduces strictly to the non-cooperative \gls{acr:ndg} defined in Definition~\ref{def:NDG}.
In addition, the two-stage formulation offers distinct advantages:
\begin{enumerate}
    \item \textbf{Resource Pooling:} 
    Unlike the strategic local optimization in \cref{eq:ndg}, where the decision space and objectives are confined to regional subnetworks, Co-investment extends the design scope to the overall network. This mitigates myopic, region-centric planning and allows the optimization to target system-wide efficiency. In addition, Co-investment relaxes budget restriction by pooling resources ($\sum b_i \leq \sum \beta_i B_i$). This enables the financing of high-cost, high-impact projects that would be financially infeasible for a single operator acting alone.

    \item \textbf{Pareto Improvement:} 
    The payoff-sharing mechanism explicitly seeks Pareto-superior outcomes over the non-cooperative \gls{acr:ne}. The framework guarantees that the resulting network configuration is at least as beneficial to every stakeholder as the purely competitive outcome.
\end{enumerate}
\end{remark}

\subsection{Design Stage 1: Non-cooperative \gls{acr:ndg}} \label{sec:Design Stage 1}
In the first design stage, operators act independently, without negotiation or coordination. 
They engage in a non-cooperative \gls{acr:ndg}, as defined in Definition~\ref{def:NDG}, where each operator optimizes its own network investments.

The decision-making process for operator $i$ is modeled by the optimization problem formulated in \cref{problem:noncoop}, with the decision variable specifically denoted as $h^{S1}_i$ to represent the strategy of operator $i$ in Stage 1.
Since operators reserve a portion of their resources for the potential cooperative stage, the original budget constraint \eqref{eq:opt_budget} is replaced by the following reduced-budget constraint:
\begin{align}
    b_i(h^{S1}_i) \leq (1-\beta_i) \budgeti, \label{eq:opt_s1_budget}
\end{align}

\noindent where $(1-\beta_i) \budgeti$ represents the funds explicitly allocated to the non-cooperative stage.
The remaining constraints remain identical to those in \cref{problem:noncoop}.

We denote the solution to this game as the Stage 1 equilibrium profile $h^{\mathrm{S1}} = (h^{\mathrm{S1}}_1, h^{\mathrm{S1}}_2)$. This profile generates the equilibrium network configuration $\graph^{\mathrm{S1}}$, which serves as the design starting point for the co-investment in Stage 2.

\subsection{Design Stage 2: Co-investment with Payoff-sharing}\label{sec:Design Stage 2}
In Stage 2, operators may negotiate joint projects, pooling portions of their budgets to fund and design \gls{acr:pt} services.
We propose a cooperative co-investment with payoff-sharing mechanism that (i) enables operators to jointly design and finance network expansions, and (ii) allocates the cooperative gains among participants.
The co-investment step seeks a network design that maximizes the \emph{sum} of operator payoffs.
The payoff-sharing steps distribute these gains so that no operator is worse off than in the non-cooperative outcome.

\subsubsection{Co-investment Optimization}

The co-investment focuses on maximizing system-wide benefits through joint optimization.
Let~$h^{\mathrm{S2}}= (h_1^{\mathrm{S2}}, h_2^{\mathrm{S2}}) \in \mathcal{H}$ denote the joint strategy profile in Stage 2, where each component $h_i^{\mathrm{S2}}$ represents the design decisions (construction and capacity) specific to the subnetwork controlled by operator $i$.
The co-investment problem ~$\mathsf{Col}$ maximizes the aggregated payoffs of both operators, subject to the pooled budget and the network constraints:
\begin{subequations} \label{problem:coop_coinvest}
\begin{alignat}{2}
    \max_{h^{\mathrm{S2}} \in \mathcal{H}} \quad 
    & \sum_{i \in \mathcal{I}} f_i\tup{h_i^{\mathrm{S2}}, y} 
    & \qquad & \text{\footnotesize (see \cref{eq:f})} \label{eq:opt_obj_s2}\\
\text{s.t.} \quad 
    & y = \demandmodel\bigl(h^{\mathrm{S2}}, \graph^{\mathrm{S1}}\bigr) 
    & & \text{\footnotesize (see \cref{eq:y})} \label{eq:opt_demand_s2}\\
    & \sum_{i \in \mathcal{I}} b_i(h_i^{\mathrm{S2}}) \leq \sum_{i \in \mathcal{I}} \beta_i \budgeti 
    & & \text{\footnotesize (see \cref{eq:cost})} \label{eq:co_budget}\\
    & \graph^{\mathrm{S2}} = \mathcal{T}(\graph^{\mathrm{S1}}, h)
    & & \text{\footnotesize (see \cref{eq:state_transition})} \label{eq:co_state_transition} \\
    & h_i^{\mathrm{S2}} = \{(d_e, s_e)\}_{e \in \edges_P^i} \nonumber \\
    & d_e \in \{0,1\}, \ s_e \in [0, s_{\max}] 
    & & \forall e \in \edges^i \label{eq:opt_vars_s2}
    %& x_e^\mathrm{S1} \leq x_e^\mathrm{S2}, \quad c_e^\mathrm{S1} \leq c_e^\mathrm{S2} 
    %& & \forall e \in \edges \label{eq:monotonicity} \
    %& & \text{Eqs.}\tup{\ref{eq:state_transition}} \nonumber,
\end{alignat}
\end{subequations}
\noindent where the objective \eqref{eq:opt_obj_s2} sums the individual payoffs of both operators.
Constraint~\eqref{eq:co_budget} ensures that the total cost across regions does not exceed the pooled cooperative funds.
%Constraints~\eqref{eq:monotonicity} ensure that the design in Stage 2 builds upon the network from Stage 1.
Equations $\tup{\ref{eq:co_state_transition}}$ ensures that the design decision in Stage 2 ($h^{\mathrm{S2}}$) builds upon the network from Stage 1 ($\graph^{\mathrm{S1}}$).

For further analysis, we define the \textit{co-investment ratio (CIR)} to represent the proportion of the total design budget allocated to co-investment
(${\sum_{i \in \mathcal{I}} \beta_i B_i}/{B}$). 
% and \textit{return on co-investment (RoC)}  is used to represent the overall profitability of the co-investment mechanism ($\Delta F^\mathrm{co}/{\sum_{i \in \mathcal{I}, t \in \mathcal{T}} \beta_i^t B_i^t}$), where $\Delta F^\mathrm{co}$ is the additional payoff from the two-stage network design.
%($\left(F^\mathrm{co}-F^\mathrm{NE}-B\right)/{B}$)

\begin{remark}[Problem complexity]
Problem (\ref{problem:coop_coinvest}) is an \gls{acr:minlp}, with binary decision variables $\set{d_e}_{e \in \edges_P}$ and continuous decision variables $\set{s_e}_{e \in \edges_P}$. 
The problem involves $\mathcal{O}(|\edges_P|)$ decision variables and $\mathcal{O}(|\edges_P|+|\mathcal{M}|)$ constraints, where $|\edges_P|$ denotes the number of \gls{acr:pt} edges in the overall \gls{acr:pt} network. 
$|\mathcal{M}|$ is the total number of trips.
\end{remark}

\subsubsection{Payoff-sharing Optimization}
To allocate the benefits generated from the cooperative network design, we propose a mechanism based on the \gls{acr:nbs}~\cite{nash1950bargaining}. This approach ensures a fair and efficient distribution of the surplus while respecting the individual rationality of the operators.

Consider a standard bargaining problem where a set of agents $\mathcal{I}$ must agree on how to split a total shareable payoff $S \in \nonnegativenumbers$. Let $v \in \nonnegativenumbers^{|\mathcal{I}|}$ denote the vector of final payoffs. The negotiation is constrained by a disagreement point $\varphi$, which represents the minimum guaranteed payoff each agent receives if negotiations fail. 
The \gls{acr:nbs} can be obtained by solving the following optimization problem:
\begin{subequations}
\label{opt:payoffsharing_standard}
\begin{align}
    \max_{v_i \in \nonnegativenumbers} & \prod_{i \in \mathcal{I}} (v_i - \varphi_i)^{\alpha_i} \\ \text{s.t.} &\quad v_i \geq \varphi_i \label{eq:dis} \\
    &\sum_{i \in \mathcal{I}} v_i = S,
\end{align}
\end{subequations}
\noindent
where $\alpha_i$ represents the bargaining power of operator $i$. The objective is to maximize the Nash Product, defined as the product of the surplus utilities over the disagreement point.
%\cathy{Question about Definition 5, the definitions is not particularly related to the disagreement payoff $\phi$, should there be another criteria on $\phi$? No need to fix this now, I'm debating whether we should delete definition anyway, let's discuss on Thrusday}

%\cathy{I would probably just say, the agents are sharing the total shareable payoff $S$ with disagreement payoofs $\phi$, then state that the Nash bargaining solution under this setting is Eq(15), then say that in our setting, the shareable payoff is ... and the disagreement point is ..., substuting this into (15) we get (18)}
% \begin{definition}[Feasible Agreement]
% \label{def:feasible}
% Let~$v_i$ be the operator~$i$'s payoff under the cooperative mechanism, and $\varphi_i$ denote the payoff for operator $i$ at the Nash equilibrium of a full-budget non-cooperative \gls{acr:ndg}.
% An agreement is \emph{feasible} if it satisfies the condition of \emph{individual rationality} for all operators:
% \begin{equation}
%     v_i \geq \varphi_i, \quad \forall i \in \mathcal{I}.
% \end{equation}
% \end{definition}
For our specific problem, we map the general parameters $\varphi$ and $S$ to the outcomes of the two-stage \gls{acr:ndg}.\\
\emph{Disagreement Point ($\varphi$):}
We define the disagreement point as the payoffs of the fully non-cooperative scenario. If the payoff-sharing mechanism fails to yield an agreement, operators revert to non-cooperative behavior, utilizing their full budgets individually rather than participating in the two-stage process. 
Thus, $\varphi_i$ corresponds to the objective value achieved by operator $i$ in the fully non-cooperative \gls{acr:ndg}. 
Formally, $\varphi_i$ is equal to the objective value achieved by operator $i$ in the equilibrium state $h^{\mathrm{NE}}$ of the game where every operator solves the problem defined in \cref{problem:noncoop} using their full budget $\budgeti$ (i.e., setting $\beta_i=0$).
\begin{equation}
    \varphi_i = f_i\bigl(h^{\mathrm{NE}}, \demandmodel(h^{\mathrm{NE}}, \graph)\bigr), \quad \forall i \in \mathcal{I}. \nonumber
\end{equation}
This help to ensure individual rationality: for any agreement to be acceptable, the final payoff must satisfy $v_i \geq \varphi_i$.\\
\emph{Shareable Payoff ($S$):}
The total value available for sharing is defined as the surplus generated by the co-investment in the second stage. Specifically, the operators retain the benefits generated from Stage 1, and the "shareable" portion is the incremental gain realized in Stage 2.
To calculate this, we first define the payoffs for operator $i$ in each stage using their objective function $f_i$ in \cref{eq:f}.
Let $F_i^{\mathrm{S1}}$ denote the payoff for operator $i$ resulting from the Stage 1, and $F_i^{\mathrm{S2}}$ denote the payoff from the Stage 2 :
\begin{align}
F_i^\mathrm{S2} &= \sum_{i \in \mathcal{I}}f_i(h^\mathrm{S_2} ,\demandmodel(h^\mathrm{S_2}, \graph_\mathrm{S_1})),  \nonumber\\
F^{\mathrm{S1}}_i &= \sum_{i \in \mathcal{I}}f_i(h^\mathrm{S_1} ,\demandmodel(h^\mathrm{S_1}, \graph)).\nonumber
\end{align}
Let $Q_i$ represent the contribution of operator $i$ to this surplus pool, calculated as the difference between the Stage 2 payoff ($F^{S2}_i$) and the Stage 1 baseline ($F^{S1}_i$), adjusted for implementation costs ($b_i$).
\begin{align}
Q_i = F_i^{\mathrm{S2}} - \tup{F_i^{\mathrm{S1}} + b_i(h^{\mathrm{S1}}_i)}. \nonumber
\end{align}
The total shareable payoff is the sum of these contributions ($S = \sum Q_i$).

Substituting these specification into the problem (\ref{opt:payoffsharing_standard}), we formulate the payoff allocation problem. Let $q \in \mathcal{Q} \subseteq \mathbb{R}^{|\mathcal{I}|}$ denote the allocation decisions. 
The final payoff for operator $i$ is the sum of their Stage 1 and their allocated share, i.e., $v_i = F^{\mathrm{S1}}_i + q_i$. The optimization problem is:
\begin{subequations} 
\label{problem:coop_payoffshare}
\begin{align}
     \max_{q_i \in \mathcal{Q}_i} \quad &  \prod_{i \in \mathcal{I}} \tup{v_i - \varphi_i}^{\alpha_i} \label{ps:obj}\\
    \text{s.t.} \quad
     & \sum_{i \in \mathcal{I}} q_i = \sum_{i \in \mathcal{I}}  Q_i \label{eq:qQ} \\
     & Q_i = F_i^{\mathrm{S_2}} - \tup{F_i^{\mathrm{S_1}} + b_i\tup{h^
     \mathrm{S_1}_i}} \label{eq:Q}\\
     & v_i = q_i + F^{\mathrm{S1}}_i,  \quad \forall i \in \mathcal{I} \label{eq:vq} \\
     &  v_i \geq \varphi_i, \quad \forall i \in \mathcal{I}, \label{eq:fagree} 
\end{align}
\end{subequations}
\noindent
where constraints~(\ref{eq:qQ}) ensure that the total shared payoff allocated across all operators is equal to the collectively agreed shareable value. 
$\alpha_i$ represents the bargaining power of operator $i$. 
For symmetric bargaining power,~$\alpha_i=1$. 
%We discuss the asymmetric bargaining power later below.

In modeling the payoff-sharing process, it is essential to account for how real-world negotiations typically unfold.
To this end, we incorporate two practical considerations that can significantly influence the eventual allocation of cooperative gains: \emph{bargaining power} and \emph{selective sharing} behavior.

\paragraph{Discussion on Bargaining Power and Exploration} 
\label{sec:Discussion on Bargaining Power and Exploration}
First, the payoff allocation can be shaped by the relative bargaining strength of each operator.
In practice, an operator's bargaining position is often tied to its level of financial commitment: those who contribute a greater proportion of the total co-investment generally wield greater influence in negotiations and, consequently, command a proportionally larger share of the cooperative benefits.
To formalize this relationship, define the bargaining power~$\alpha_i$ of operator~$i$ as:
\begin{equation}
     \alpha_i = \frac{\beta_i B_i}{\sum_{j \in \mathcal{I}} \beta_j B_j} , \quad \forall i \in \mathcal{I}. \label{eq:alpha}
\end{equation}
This formulation ensures that the influence of each operator in determining the final payoff allocation is directly proportional to its contribution to the joint investment pool.

Second, in many real-world contexts, primary investors may be willing to share only the surplus value that is generated beyond their own operational network, keeping the internally generated benefits for themselves.
%\cathy{I have a get feeling that the problem is still equivalent to (18) even though we introduce $\epsilon_i$, let's discuss about this.}
To capture such selective sharing, we introduce a binary parameter~$\epsilon_i \in \{0,1\}$ representing operator~$i$'s willingness to share the cooperative surplus.
This parameter modifies both the total shareable payoff~$q_i$ and the operator's resulting benefit~$v_i$, as follows:
\begin{subequations}
    \begin{align}
    &\sum_{i \in \mathcal{I}} q_i = \sum_{i \in \mathcal{I}} \epsilon_i Q_i, \label{eq:epsilon_qQ} \\
     &v_i = q_i + F^{S1}_i + \tup{1 - \epsilon_i} Q_i, \quad \forall i \in \mathcal{I}, \label{eq:epsilon_vq} 
    \end{align}
\end{subequations}
Here,~$\epsilon_i = 1$ denotes full willingness to share the surplus generated within one's own network, whereas~$\epsilon_i = 0$ indicates that the operator retains this portion exclusively.
\mh{We note that in the public transit setting, $\epsilon_i$ is a term of the cooperation agreement, typically set through regulation or inter-municipal coordination: an authority protecting a smaller municipality would mandate full sharing ($\epsilon_i = 1$). 
Absent such coordination, $\epsilon_i$ reflects relative bargaining positions, and a stronger operator may retain its internal surplus ($\epsilon_i = 0$), exposing the weaker operator to the exploitation analyzed in \cref{sec:Zurich-Winterthur Network}.}
Both bargaining weights and selective sharing can be incorporated into the payoff-sharing framework by modifying \eqref{ps:obj}, \eqref{eq:qQ}, and \eqref{eq:vq}.
Their implications are examined in the case study in \Cref{sec:Zurich-Winterthur Network}.
Finally, to facilitate interpretation of the results, we introduce two conceptual benchmarks.\\
\emph{Minimum Guaranteed Return (MGR)}: 
The smallest guaranteed payoff increases once an operator's co-investment ratio exceeds $\beta_\mathrm{MGR}$:
\[
\mathrm{MGR}_i = \min_{\beta_i \geq \beta_{\mathrm{MGR}}} \frac{v_i(\beta_i) - \varphi_i}{\varphi_i}.
\]
It represents the return security under the cooperative arrangement.\\
\emph{Strategic Exploitation Threshold (SET)}: The co-investment ratio beyond which an operator's marginal gains decline due to others' strategic actions (e.g., withholding, benefit reallocation):
\[
\frac{\partial v_i}{\partial \beta_i} < 0, \quad \forall \, \beta_i > \mathrm{SET},
\]
where extra co-investment no longer yields more returns.

% \begin{remark} [Interdependence of design stages]
%     We model network design as a two-stage process: a non-cooperative \gls{acr:ndg} followed by co-investment and payoff sharing.
%     % When applied repeatedly over the long term, these stages form a complex multi-stage dynamic. 
%     This paper focuses on the mechanism and equilibria, while future research should examine the dynamic optimality that emerges from repeated interaction between these stages.
% \end{remark}
\mh{
\begin{remark}[Dynamic and strategic extensions]
We model network design as a two-stage process comprising a non-cooperative \gls{acr:ndg} followed by co-investment and payoff sharing. In our current formulation, both the co-investment ratio $\beta_i$ and the willingness-to-share parameter $\epsilon_i$ are treated as fixed commitments. 
The value of $\beta_i$ impacts the model through both the budget split and the bargaining weight $\alpha_i$ (\cref{eq:alpha}). 
The value of $\epsilon_i$ impacts the total shareable payoff. 
Allowing operators to endogenously determine $\beta_i$ and $\epsilon_i$ would yield a broader meta-game, in which each operator strategically selects their parameters based on anticipated Stage-2 bargaining payoffs. We leave the formal analysis of this meta-game, along with the dynamic optimality emerging from repeated interactions, to future work.
\end{remark}
}

\section{Theoretical Analysis}\label{sec:theoretical}
In Sections \ref{sec:Game-theoretical Framework} and \ref{sec:coinvestment}, we established the frameworks for the non-cooperative and cooperative network design, respectively. 
To ensure the proposed framework with specific objective functions, demand models, budget constraints, and decision variables is both computationally tractable and theoretically sound, we now analyze its fundamental properties. Specifically, we establish conditions under which the individual operator's strategic optimization problem is solvable (Lemma~\ref{lemma:MICP}), demonstrate that a stable equilibrium exists for the relaxed non-cooperative game (Proposition~\ref{prop:existence}), and prove that the cooperative payoff-sharing mechanism yields a unique, valid solution (Proposition~\ref{prop:payoff}).

\subsection{Computational Tractability}
We first study the properties of the strategic local optimization problem~$\mathsf{Loc}_i$ defined in \cref{problem:noncoop}.
Its complexity depends on the objective function and the demand model.
To assess solvability, we begin by defining the continuous relaxation of the game.
\begin{definition}[Continuous Relaxation of the \gls{acr:ndg}]
The continuous relaxation of the non-cooperative game $\networkgame$, denoted by $\networkgame_{\mathrm{cont}}$, is obtained by replacing the discrete strategy space $\mathcal{H}_i$ with its convex hull $\bar{\mathcal{H}}_i$. Specifically, the binary decision variables $d_e \in \{0,1\}$ are relaxed to $\bar{d_e} \in [0,1]$.
\end{definition}
\begin{lemma}[Convexity]
\label{lemma:MICP}
The continuous relaxation of the strategic local optimization problem (Eq. 12) is a convex optimization program, and hence the original problem in \cref{problem:noncoop} is a mixed-integer convex program (MICP), provided that the following conditions hold for all \gls{acr:pt} edges $e \in \edges^i_P$:
\begin{enumerate}
    \item \gls{acr:pt} utility is greater than or equal to the alternative:
    \begin{align}
        u_m^A - u_m^P \leq 0 \label{eq:condition_u}
    \end{align}
    \item The marginal gain from shifting users to \gls{acr:pt} is non-negative:
    \begin{align} 
    \Delta_e &= - l_e
    \left(
            \weighte \gamma_{2}^\mathrm{P} 
            + \weightc \left(\frac{\gamma_\mathrm{vot}}{v_P} + \gamma_{1}^\mathrm{P}\right)
            -
            \weightp \gamma_{1}^\mathrm{P}
    \right) \nonumber
    \\ 
    & + \sum_{a\in \edges^A_e} 
    l_a
    \left( 
            \weighte \gamma_{2}^\mathrm{A} 
            + \weightc \left(\frac{\gamma_\mathrm{vot}}{v_A}+\gamma_{1}^\mathrm{A}\right)
      \right) 
    \geq 0.  \label{eq:condition_e}
    \end{align}
\end{enumerate}
\end{lemma}
%\cathy{Once we put Definition 6 before Lemma 1, we can state the Lemma as the continuous relatxation of (12) is a convex program and the original problem is MICP} 

\begin{proof}
An optimization problem is a \gls{acr:micp} if its continuous relaxation is convex. We analyze the components of Problem (\ref{problem:noncoop}):

First, with the continuous relaxation of $d$, the decision variables lie in a compact and bounded domain, with $\bar{d}\in [0,1]^{|\edges^i_P|}$ and $s \in [0, s_{\max}]^{|\edges^i_P|}$.

Next, we examine the intermediate variables for route choice $p_m$ (in \cref{eq:p_m}) and edge flow $y$ (in \cref{eq:y}). 
Travel utility $u^P_m$ is affine in $d$, as it varies linearly and monotonically with $x$ and $d$, and the alternative-route utility $u^A_e$ is independent of $d$. Therefore, the utility difference $u^P_m - u^A_m$ is affine in $d$. 
Under this condition, the sigmoid function $p_m = 1/\tup{1 + e^{\eta}}$, with $\eta=u_m^A - u_m^P \leq 0$ (Condition (\ref{eq:condition_u})), is concave and monotonically increasing in $d$. Thus, the potential \gls{acr:pt} demand $\tilde{y}_e$ is concave.  

The objective function in \cref{eq:f} accounts for regional emissions, social welfare, and revenue.
It can be reformulated by grouping terms associated with \gls{acr:pt} flow and alternative-mode flow:
\begin{align}
    f_i & = \sum_{e \in \edges^i_P} K^P_e y_e + \sum_{a \in \edges^i_A} K^A_a y_a - \weightp b_i(h_i) \nonumber \label{eq:f_reformulate1}
\end{align}
where $b_i(h_i)$ represents linear construction and operation costs; $K^P_e$ and $K^A_a$ represent the marginal benefit coefficients for \gls{acr:pt} and alternative flows, respectively:
\begin{align}
    K^P_e &= - l_e \left( \weighte \gamma_{2}^\mathrm{P} + \weightc \left(\frac{\gamma_\mathrm{vot}}{v_P} + \gamma_{1}^\mathrm{P}\right) - \weightp \gamma_{1}^\mathrm{P} \right), \nonumber \\
    K^A_a &= - l_a \left( \weighte \gamma_{2}^\mathrm{A} + \weightc \left(\frac{\gamma_\mathrm{vot}}{v_A} + \gamma_{1}^\mathrm{A}\right) \right), \nonumber 
\end{align}
Note that in Condition (\ref{eq:condition_e}), $\Delta_e = K^P_e - \sum K^A_a \geq 0$.

To analyze the convexity, we further decompose the objective function into PT-edge contributions $g_e$:
\begin{align}
    f_i = \sum_{e \in \edges^i_P} g_e - \weightp b_i(h_i). \nonumber
\end{align}
The expression for $g_e$ depends on the relationship between the potential \gls{acr:pt} flow $\tilde{y}_e$ and the capacity $c_e$.
Let $\edges^A_e$ denote the set of alternative edges as a substitute of \gls{acr:pt} edge $e$. As defined before, let $\tilde{y}_e=\sum_{m \in \mathcal{M}} \mathds{1}_{e \in \edges^P_m} \alpha_m p_m$ denote the potential \gls{acr:pt} flow intending to use edge $e$.

\textit{Case 1 ($\tilde{y}_e \le c_e$)}: Edge $e$ accommodates all potential \gls{acr:pt} demand, and the remaining demand uses the alternative edges.
The contribution is:
\begin{align}
    g_e &= K^P_e \tilde{y}_e + \sum_{a\in \edges^A_e} K^A_a \left(\sum_{m \in \mathcal{M}} \mathds{1}_{e \in \edges^P_m} \alpha_m - \tilde{y}_e\right) \nonumber \\
    &= \underbrace{\sum_{a\in \edges^A_e} K^A_a \left(\sum_{m \in \mathcal{M}} \mathds{1}_{e \in \edges^P_m} \alpha_m\right)}_{\text{Constant}} + \underbrace{\left(K^P_e - \sum_{a\in \edges^A_e} K^A_a\right)}_{\Delta_e \geq 0} \tilde{y}_e. \nonumber
\end{align}
Since $\Delta_e \geq 0$ (Condition (\ref{eq:condition_e})) and $\tilde{y}_e$ is concave, the second term is concave.

\textit{Case 2 ($\tilde{y}_e > c_e$)}: The \gls{acr:pt} flow is capped at $c_e$, and the excess demand shifts to the alternative, then:
\begin{align}
    g_e &= K^P_e c_e + \sum_{a\in \edges^A_e} K^A_a \left(\left(\sum_{m \in \mathcal{M}} \mathds{1}_{e \in \edges^P_m} \alpha_m - \tilde{y}_e\right) +(\tilde{y}_e - c_e) \right) \nonumber \\
    &= K^P_e c_e + \sum_{a\in \edges^A_e} K^A_a \left( \sum_{m \in \mathcal{M}} \mathds{1}_{e \in \edges^P_m} \alpha_m - c_e \right). \nonumber
\end{align}
Here, the variable terms $\tilde{y}_e$ cancel out, and $g_e$ is linearly related to the capacity decision $c_e$.

Therefore, provided that for each $e \in \edges^P$, $\Delta_e$ is non-negative, the objective $f_i$ is a sum of concave, linear, and constant terms. By the composition rules for convexity, maximizing a concave objective over a linear domain constitutes a convex optimization problem. Therefore, the original problem is an \gls{acr:micp}.
\end{proof}
\begin{remark} [Illustration of convexity conditions]
This result implies that the strategic network design problem in a multi-agent environment remains convex if the service is designed to be beneficial for both users and operators:
\begin{itemize}
    \item For users, Condition \eqref{eq:condition_u} implies that the \gls{acr:pt} service is more attractive to users than the alternative, rendering the route choice function concave.
    \item For operators, Condition \eqref{eq:condition_e} requires that the net marginal benefit of accommodating a user on \gls{acr:pt} edge $e$, accounting for profit, emissions, and social welfare, is non-negative ($\Delta_e \geq 0$). This ensures that shifting demand to \gls{acr:pt} improves the objective, even when capacity constraints force excess demand back to alternative modes.
\end{itemize}
From a computational perspective, identifying the problem as an \gls{acr:micp} is significant because it enables global optimality guarantees via standard methods such as branch-and-bound or outer-approximation algorithms~\cite{Cauligi9304043,lee2012mixed}.
\end{remark}

\subsection{Existence of Equilibrium}
With the tractability of the individual operator's problem, we then address the system stability.
Since the existence of pure Nash Equilibria in discrete games is not guaranteed, we analyze the properties of the continuous relaxation.
To prove the existence of an equilibrium in $\networkgame_{\mathrm{cont}}$, we invoke two fundamental theorems from fixed-point theory.

\begin{theorem}[Kakutani’s Fixed Point theorem\cite{kakutani1941generalization}]
\label{thm:kaku}
    Let $ \Phi: X \to 2^X $ be a set-valued function on $X$. There exists $x^* \in \Phi\tup{x^*}$ if the following conditions hold:
\begin{enumerate}
    \item $ X $ is a nonempty, compact, and convex subset of a Euclidean space.
    \item For all $x \in X $, $\Phi\tup{x}$ is nonempty, convex, and compact.
    \item The graph $\{\tup{x, y} \in X \times X : y \in \Phi\tup{x}\} $, is closed.
\end{enumerate}
\end{theorem}

\begin{theorem}
[Maximum Theorem\cite{berge1963topological}] 
\label{thm:max}
Let  $K \subset \mathbb{R}^n$  be compact, and let $ f : K \times Y \to \mathbb{R} $ be a map that is continuous on $ K \times Y $ and convex in $ K $ for each fixed $ y \in Y $. 
Then, for $ y \in Y $, $ \phi\tup{y} = \arg \max_{x \in K} f\tup{x, y} $ is upper-hemicontinuous, and $ \phi\tup{y} \subset K $ is compact and convex.
\end{theorem}

\begin{proposition} [Existence of Pure NE] \label{prop:existence}
If \eqref{eq:condition_u} and \eqref{eq:condition_e} hold, the game $\networkgame_{\mathrm{cont}}$ possesses at least one Pure Strategy Nash Equilibrium.
\end{proposition}

\begin{proof}
According to Lemma \ref{lemma:MICP}, assuming condition~\eqref{eq:condition_u} and \eqref{eq:condition_e} hold, for each operator $i \in \mathcal{I}$, the optimization problem $\mathsf{Loc}_i$ is an \gls{acr:micp} in the decision variables $h_i \in \mathcal{H}_i$.
%where $\mathcal{H}_i$ denotes the feasible set defined by the constraints of the problem.
When the network design strategies are continuous, the feasible strategy space $\mathcal{H}_i$ is compact and convex, and the objective function $f_i\tup{h_i, y}$ is continuous in $h_i \in \mathcal{H}_i$.
Then, we define the best response correspondence for each operator $i \in \mathcal{I}$ as:
\[
\phi_{i}\tup{h_{-i}}:= \arg \max_{h_i} f_i\tup{h_i, y},
\]
which returns the set of optimal responses to the fixed strategies of other operators.
Let the joint strategy profile be denoted by $h=\{h_i\}_{i \in \mathcal{I}} \in \mathcal{H}$
and the set-valued function to be:
\[
\Phi\tup{h} := \tup{ \phi_{i}\tup{h_{-i}} }_{i \in \mathcal{I}},
\]
Based on the \cref{thm:max}, the best response correspondence $\phi_{i}\tup{h_{-i}}$ is upper hemicontinuous and has non-empty, compact, and convex values, since the objective function is continuous and the feasible set is compact.
By \cref{thm:kaku}, the set-valued function $\Phi\tup{x}$ has at least one fixed point $h^{\star} \in \mathcal{H}$, such that $h^{\star} \in \Phi\tup{h^{\star}}$.
This fixed point is a pure Nash equilibrium of the \gls{acr:ndg} in the continuous relaxation according to the \cref{def:ne}.
\end{proof} 

% \begin{remark} [Equilibrium existence]
% This result guarantees the existence of equilibrium investment profiles in the continuous limit. This provides a rigorous reference for cooperative network design, by ensuring that the disagreement point $\varphi$ for the subsequent cooperative stage, as well as the equilibrium investment profile $h^{\mathrm{S1}}$ and network $\graph^{\mathrm{S1}}$ at Stage 1, are well defined.
% \mh{We emphasize that this result pertains to the continuous relaxation and is not claimed to characterize the integer equilibria computed in the experiments. The latter are obtained directly via Iterative Best Response. }
% \end{remark}
\mh{
\begin{remark}[Scope of the existence result]
The existence result in \cref{prop:existence} is established for the continuous relaxation, and it does not characterize the integer equilibrium. 
The link to the integer setting is constructive: since each best-response subproblem is the \gls{acr:micp} of \cref{lemma:MICP} and can be solved to a prescribed optimality gap, any fixed point of the \gls{acr:ibr} procedure is a pure-strategy Nash equilibrium of the integer game (\cref{def:ne}) to within that gap. An analytical correspondence between the continuous and integer equilibria is left to future work.
\end{remark}
}
\subsection{Feasibility of Payoff Sharing}
Finally, we analyze the cooperative stage. The validity of \gls{acr:nbs} depends on the existence of a feasible agreement space.

\begin{proposition}[Existence and Uniqueness of Cooperative Solution]
\label{prop:payoff}
The payoff-sharing optimization problem (Eq. \ref{problem:coop_payoffshare}) yields a unique optimal allocation vector $v^*$ if and only if the total cooperative surplus is strictly positive:
\begin{align}
    \sum_{i \in \mathcal{I}} F_i^\mathrm{S2} - \sum_{i \in \mathcal{I}} b_i\tup{h_i^{\mathrm{S1}}} >\sum_{i \in \mathcal{I}} \varphi_i \label{eq:corollary1},
\end{align}
\end{proposition}
\noindent
This inequality shows that the payoff mechanism applies only when co-investment yields a higher total payoff than in the case of complete non-cooperation.

\begin{proof}
Condition~(\ref{eq:corollary1}) ensures that the set of feasible payoffs ($\mathcal{Q}$) is a non-empty, convex, and compact subset of $\mathbb{R}^{|\mathcal{I}|}$. 
Maximizing the Nash Product is equivalent to maximizing its logarithm, $\sum \alpha_i \ln(v_i - \varphi_i)$, which is a strictly concave function.
Since maximizing a strictly concave function over a convex compact set guarantees a unique global maximum, the solution $v^*$ exists and is unique.
\end{proof}

\section{Case Study} \label{sec:exp}
To demonstrate the effectiveness of the proposed framework, we conduct two sets of numerical experiments.
The first is based on the well-known Sioux Falls network in the United States (\cref{fig:sioux}), which serves as a benchmark for testing transportation network models~\cite{leblanc1975efficient}. 
The second applies the framework to real-world public transport networks in the Swiss cities of Zurich and Winterthur (\cref{fig:zuriwinti}).
In both studies, we model the decision-making of public transport operators with private car travel considered as the alternative mode for users.
To account for long-term planning under changing demand conditions, we assume that annual travel demand grows by a factor of~$\tau$.

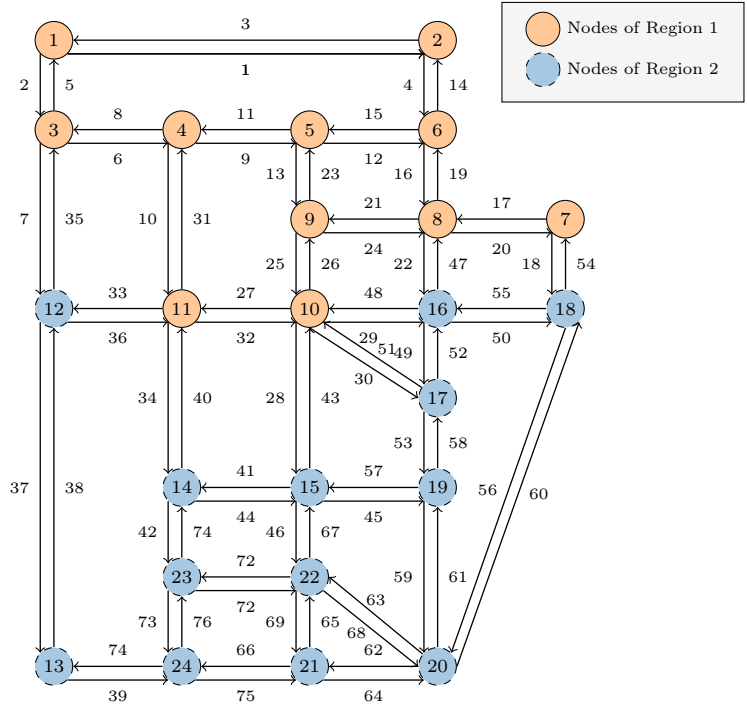
\begin{figure}[!htbp]
    \centering
    \begin{tikzpicture}[scale=1.7]
    % Nodes
    \definecolor{steelblueline}{RGB}{167,200,227}%{RGB}
    \tikzstyle{node_style1}=[circle, draw, fill={rgb,255:red,255;green,200;blue,150}, minimum size=5mm, inner sep=0pt, font=\scriptsize]
    \tikzstyle{node_style2}=[circle, draw, dashed, fill=steelblueline, minimum size=5mm, inner sep=0pt, font=\scriptsize]
    %{rgb,255:red,128;green,200;blue,255}
    \def \radius {10cm}
    \node[node_style2] (13) at (0, 0) {13};
    \node[node_style2] (24) at (1, 0) {24};
    \node[node_style2] (21) at (2, 0) {21};
    \node[node_style2] (20) at (3, 0) {20};
    \node[node_style2] (23) at (1, 0.7) {23};
    \node[node_style2] (22) at (2, 0.7) {22};
    \node[node_style2] (14) at (1, 1.4) {14};
    \node[node_style2] (15) at (2, 1.4) {15};
    \node[node_style2] (19) at (3, 1.4) {19};
    \node[node_style2] (17) at (3, 2.1) {17};
    \node[node_style2] (12) at (0, 2.8) {12};
    \node[node_style1] (11) at (1, 2.8) {11};
    \node[node_style1] (10) at (2, 2.8) {10};
    \node[node_style2] (16) at (3, 2.8) {16};
    \node[node_style2] (18) at (4, 2.8) {18};
    \node[node_style1] (7) at (4, 3.5) {7};
    \node[node_style1] (8) at (3, 3.5) {8};
    \node[node_style1] (9) at (2, 3.5) {9};
    \node[node_style1] (6) at (3, 4.2) {6};
    \node[node_style1] (5) at (2, 4.2) {5};
    \node[node_style1] (4) at (1, 4.2) {4};
    \node[node_style1] (3) at (0, 4.2) {3};
    \node[node_style1] (2) at (3, 4.9) {2};
    \node[node_style1] (1) at (0, 4.9) {1};

    % Edges with labels
    \tikzstyle{line_style}=[->,  line width=0.5pt, font=\fontsize{6}{8}\selectfont]
    \draw[line_style] (2) -- (1) node[midway, above]{3};
    \draw[line_style] (1.south east) -- (2.south west) node[midway, below] {1};
    \draw[line_style] (4) -- (3) node[midway, above]{8};
    \draw[line_style] (5) -- (4) node[midway, above]{11};
    \draw[line_style] (6) -- (5) node[midway, above]{15};
    \draw[line_style] (7) -- (8) node[midway, above]{17};
    \draw[line_style] (8) -- (9) node[midway, above]{21};
    \draw[line_style] (18) -- (16) node[midway, above]{55};
    \draw[line_style] (16) -- (10) node[midway, above]{48};
    \draw[line_style] (10) -- (11) node[midway, above]{27};
    \draw[line_style] (11) -- (12) node[midway, above]{33};
    \draw[line_style] (19) -- (15) node[midway, above]{57};
    \draw[line_style] (15) -- (14) node[midway, above]{41};
    \draw[line_style] (22) -- (23) node[midway, above]{72};
    \draw[line_style] (20) -- (21) node[midway, above]{62};
    \draw[line_style] (21) -- (24) node[midway, above]{66};
    \draw[line_style] (24) -- (13) node[midway, above]{74};
    \draw[line_style] (1.south west) -- (3.north west) node[midway, left]{2};
    \draw[line_style] (12.south west) -- (13.north west) node[midway, left]{37};
    \draw[line_style] (4.south west) -- (11.north west) node[midway, left]{10};
    \draw[line_style] (11.south west) -- (14.north west) node[midway, left]{34};
    \draw[line_style] (14.south west) -- (23.north west) node[midway, left]{42};
    \draw[line_style] (23.south west) -- (24.north west) node[midway, left]{73};
    \draw[line_style] (5.south west) -- (9.north west) node[midway, left]{13};
    \draw[line_style] (9.south west) -- (10.north west) node[midway, left]{25};
    \draw[line_style] (10.south west) -- (15.north west) node[midway, left]{28};
    \draw[line_style] (15.south west) -- (22.north west) node[midway, left]{46};
    \draw[line_style] (22.south west) -- (21.north west) node[midway, left]{69};
    \draw[line_style] (2.south west) -- (6.north west) node[midway, left]{4};
    \draw[line_style] (6.south west) -- (8.north west) node[midway, left]{16};
    \draw[line_style] (8.south west) -- (16.north west) node[midway, left]{22};
    \draw[line_style] (16.south west) -- (17.north west) node[midway, left]{49};
    \draw[line_style] (17.south west) -- (19.north west) node[midway, left]{53};
    \draw[line_style] (19.south west) -- (20.north west) node[midway, left]{59};
    \draw[line_style] (3.south west) -- (12.north west) node[midway, left]{7};

    \draw[line_style] (1.south east) -- (2.south west) node[midway, below] {1};
    \draw[line_style] (3.south east) -- (4.south west) node[midway, below]{6};
    \draw[line_style] (4.south east) -- (5.south west) node[midway, below]{9};

    \draw[line_style] (5.south east) -- (6.south west) node[midway, below]{12};
    \draw[line_style] (8.south east) -- (7.south west) node[midway, below]{20};
    \draw[line_style] (9.south east) -- (8.south west) node[midway, below]{24};
    \draw[line_style] (16.south east) -- (18.south west) node[midway, below]{50};
    \draw[line_style] (10.south east) -- (16.south west) node[pos=0.45, below]{29};
    \draw[line_style] (11.south east) -- (10.south west) node[midway, below]{32};
    \draw[line_style] (12.south east) -- (11.south west) node[midway, below]{36};
    \draw[line_style] (15.south east) -- (19.south west) node[midway, below]{45};
    \draw[line_style] (14.south east) -- (15.south west) node[midway, below]{44};
    \draw[line_style] (23.south east) -- (22.south west) node[midway, below]{72};
    \draw[line_style] (21.south east) -- (20.south west) node[midway, below]{64};
    \draw[line_style] (24.south east) -- (21.south west) node[midway, below]{75};
    \draw[line_style] (23) -- (14) node[midway, right]{74};
    \draw[line_style] (3) -- (1) node[midway, right]{5};
    \draw[line_style] (12) -- (3) node[midway, right]{35};
    \draw[line_style] (13) -- (12) node[midway, right]{38};
    \draw[line_style] (11) -- (4) node[midway, right]{31};
    \draw[line_style] (14) -- (11) node[midway, right]{40};
    \draw[line_style] (13.south east) -- (24.south west) node[midway, below]{39};
    \draw[line_style] (24) -- (23) node[midway, right]{76};
    \draw[line_style] (9) -- (5) node[midway, right]{23};
    \draw[line_style] (10) -- (9) node[midway, right]{26};
    \draw[line_style] (15) -- (10) node[midway, right]{43};
    \draw[line_style] (22) -- (15) node[midway, right]{67};
    \draw[line_style] (21) -- (22) node[midway, right]{65};
    \draw[line_style] (6) -- (2) node[midway, right]{14};
    \draw[line_style] (8) -- (6) node[midway, right]{19};
    \draw[line_style] (16) -- (8) node[midway, right]{47};;
    \draw[line_style] (17) -- (16) node[midway, right]{52};
    \draw[line_style] (19) -- (17) node[midway, right]{58};
    \draw[line_style] (20) -- (19) node[midway, right]{61};
    \draw[line_style] (18.south) -- (20.north east) node[midway, left]{56};
    \draw[line_style] (20.east) -- (18.south east) node[midway, right]{60};
    \draw[line_style] (7.south west) -- (18.north west) node[midway, left]{18};
    \draw[line_style] (18) -- (7) node[midway, right]{54};
    \draw[line_style] (20) -- (22.east) node[midway,above]{63};
    \draw[line_style] (22.south east) -- (20.west) node[pos=0.35, below ]{68};
    \draw[line_style] (17) -- (10.south east)  node[pos=0.35, above]{51};
    \draw[line_style] (10.south) -- (17.west) node[midway, below]{30};
    %legend
    \node[scale=0.9,draw, fill={gray!8!white}, anchor=north east,  font=\scriptsize, text width=3.3cm, text centered] at (5.4, 5.2) {
        \begin{tabular}{l}
            \begin{tikzpicture}
                \node[node_style1] (legendA) {};
                \node[draw=none, minimum size=3mm,fill=none, anchor=west] at (legendA.east) { Nodes of Region 1};
            \end{tikzpicture} \\
            \begin{tikzpicture}
                \node[node_style2] (legendB) {};
                \node[draw=none, minimum size=3mm,fill=none, anchor=west] at (legendB.east) { Nodes of Region 2};
            \end{tikzpicture}
        \end{tabular}
        };
\end{tikzpicture}

% \begin{figure}
%     \centering
%     \begin{tikzpicture}
%       \begin{axis}
%         [
%         boxplot/draw direction=x,
%         ylabel={Time (minutes)},
%         xlabel={Co-investment Return (CHF)},
%         height=6cm,
%         width=10cm,
%         %ymin=0,ymax=11,
%         cycle list={blue!30!white, {rgb,255:red,255;green,200;blue,150}}
%         ytick={1,2,3,4,5,6},
%         yticklabels={HJ,JK,HI,IO,OP,OL},
%         xtick={1,2,3,4},
%         xticklabels={0, 100k,200k,300k},
%         ]
%         \addplot+[
%          fill={rgb,255:red,255;green,200;blue,150},
%          fill opacity=0.2,
%         boxplot prepared={
%           median=2.59,
%           upper quartile=3.35,
%           lower quartile=2,
%           upper whisker=4.4,
%           lower whisker=1.1
%         },
%         ] coordinates {};
%         \addplot+[
%          fill=blue!30!white,
%          fill opacity=0.2,
%         boxplot prepared={
%           median=6.57,
%           upper quartile=7.93,
%           lower quartile=5.93,
%           upper whisker=10.68,
%           lower whisker=4.03
%         },
%         ] coordinates {};
%       \end{axis}
%     \end{tikzpicture}
%     \caption{Caption}
%     \label{fig:enter-label}
% \end{figure}
    \caption{Sioux Falls network, subdivided between Region 1 and 2.}
    \label{fig:sioux}
\end{figure}
\begin{figure}[!htbp]
    \centering
    \includegraphics[width=0.5\linewidth]{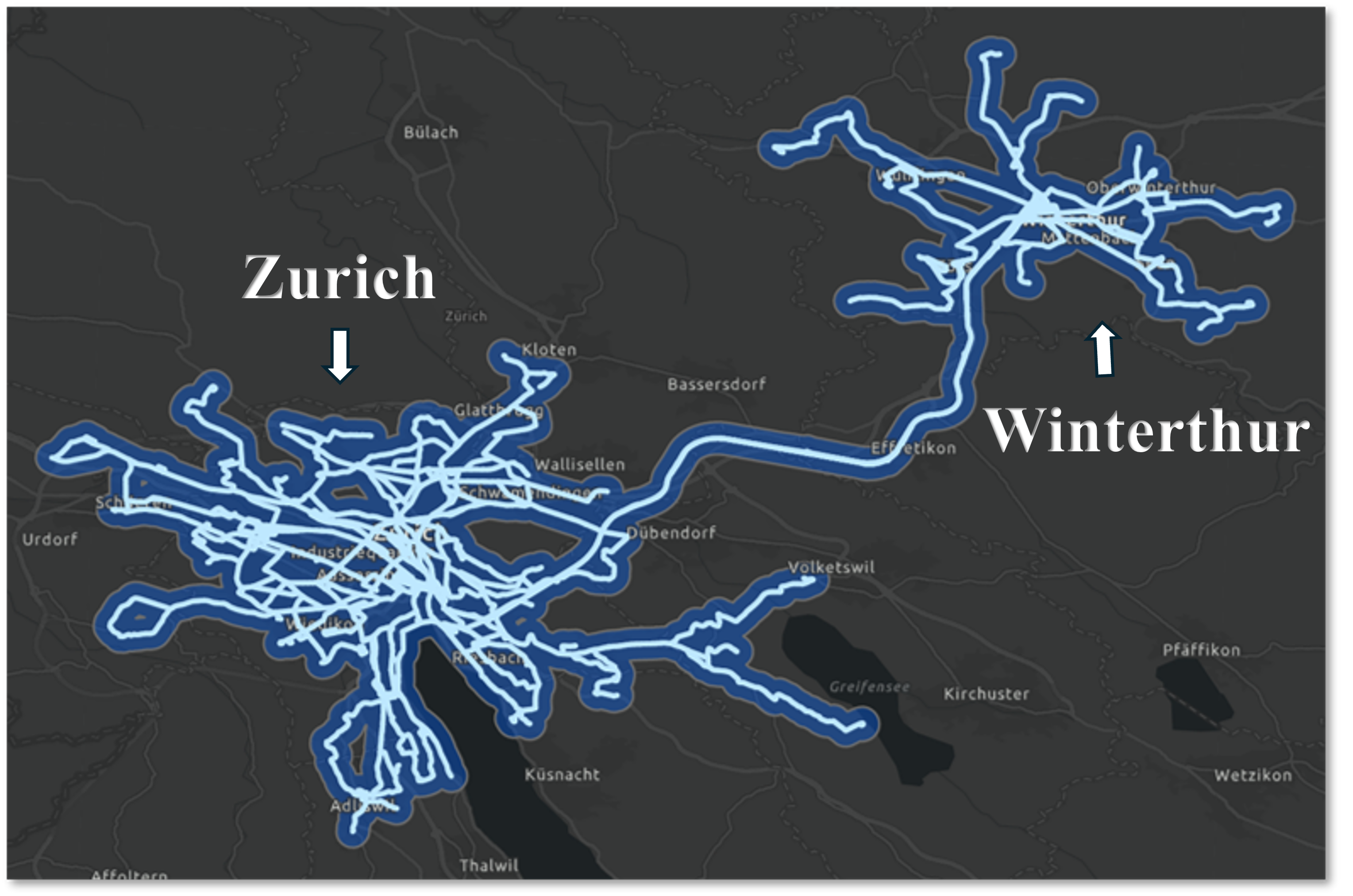}
    \caption{Public transport networks in Zurich and Winterthur.}
    \label{fig:zuriwinti}
\end{figure}

\subsection{Sioux Falls Network}
The Sioux Falls network consists of 24 nodes and 76 edges, partitioned into two regions: Region 1 (nodes 1–11), and Region 2 (nodes 12–24). 
We begin by examining a baseline scenario in which both regions are homogeneous in terms of construction resources and travel demand distributions.
We then simulate a three-year network design horizon in which, at the start of each year, regional operators decide whether to engage in co-investment and, if so, determine the co-investment ratio.
The baseline for comparison is a fully non-cooperative network design game, in which no negotiation occurs between agents ($\beta_i = 0$).
This allows us to quantify the system-wide improvements enabled by the proposed cooperative mechanisms. We use the model parameters for the Sioux Falls network case in~\cite{he25}.

\subsubsection{System improvement from co-investment} \label{sec:System Improvement from Co-investment}
\Cref{fig:overview} presents equilibrium solutions under different co-investment strategies. 
% The horizontal axis, ``years of cooperation,'' indicates how many years the regional authorities choose to adopt the proposed mechanisms.
% System improvement is measured in CHF/day\footnote{For uniformity, we adopt the CHF currency for all the case studies. At time of submission, 1.0 CHF is equivalent to 1.24 USD.}{} relative to the baseline~$(0, 0, 0)$, a solution with no co-investment and no payoff-sharing.
% To assess how close each solution comes to the system-optimal (centralized) design, we compute the percentage of environmental, social, and economic gains achieved.
% Specifically, two illustrative solutions are highlighted.
% The filled star indicates the decision to co-invest during the network design phase. 
% The presented networks highlight the differences between the specific cases and the baseline network: edges in red indicate more resources allocated, while green edges represent fewer constructed edges.
The horizontal axis (``years of cooperation'') denotes the duration operators adopt the proposed mechanisms. 
System improvement, in CHF/day\footnote{For uniformity, we adopt the CHF currency for all the case studies. At the time of submission, 1.0 CHF is equivalent to 1.24 USD.}, is measured relative to the baseline $(0,0,0)$ with no co-investment or payoff-sharing. 
We also report the percentage of environmental, social, and economic gains toward the system-optimal design. 
Specifically, two illustrative solutions are highlighted.
The filled star indicates the decision to co-invest during the network design phase, with red edges indicating added resources and green edges indicating fewer constructed edges compared to the baseline.
The red dot represents the Highest-Return Network, which indicates the equilibrium network with the highest improvement in the objective. 
This solution requires continuous co-investment in each design year, with the co-investment budget accounting for 50\% of the total budget. As a result, emissions can be reduced by 12.1 tons/day, revenue increases by 19.6k CHF/day, and customer costs are reduced by 28.8k CHF/day. 
With the proposed mechanism, allocating 50\% of the budget for co-investment can achieve outcomes that are very close to the optimal system results, reaching 96\%, 96\%, and 100\% across the three dimensions.
The blue dot represents the solution where the design strategy yields the most investment-efficient outcome. By co-investing only 3.3\% of the total budget in the initial design year, the system achieves a 3.7 ton/day reduction in emissions, 8.8k CHF/day in travel cost savings, and an additional revenue of 6.7k CHF/day.
\begin{figure*}[!htbp]
    \centering
    \includegraphics[width=\linewidth]{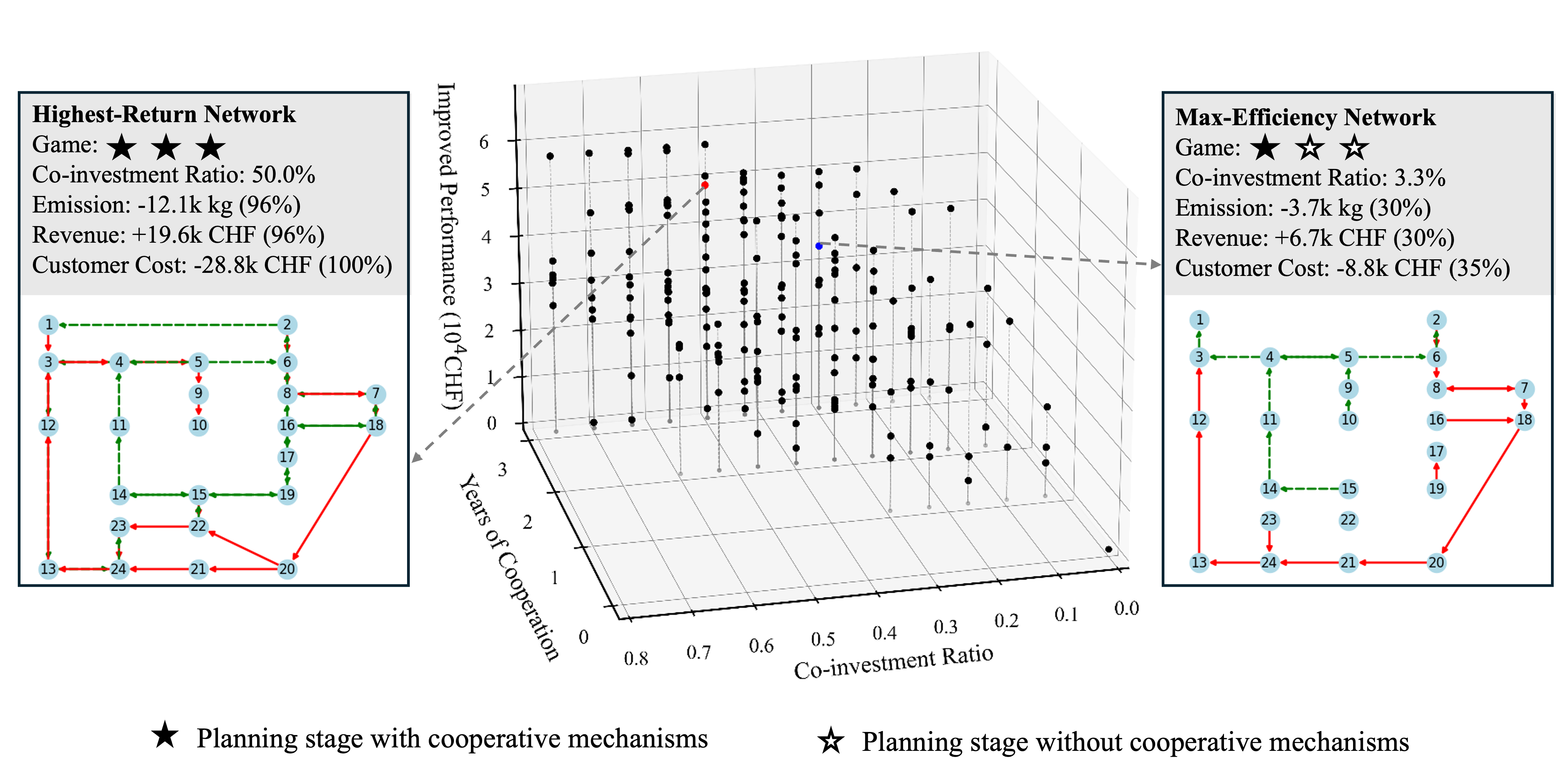}
    \caption{Equilibrium solutions of interactive network design.}
    \label{fig:overview}
\end{figure*}

% It can be observed from \Cref{fig:co_ratio} that, with the same co-investment ratio, the performance varies depending on the year of investment and the distributed funding. 
As shown in \Cref{fig:co_ratio}, performance varies with the timing and allocation of investment, even under the same co-investment ratio.
For instance, by allocating 30\% of the budget in the first year and 10\% in the second year for co-investment, the resulting network can approach the system-optimal solution by 86\% with an additional return of 51k CHF/day, exceeding other budget distribution strategies with the same total investment.
\begin{figure}[!htbp]
    \centering
    \input{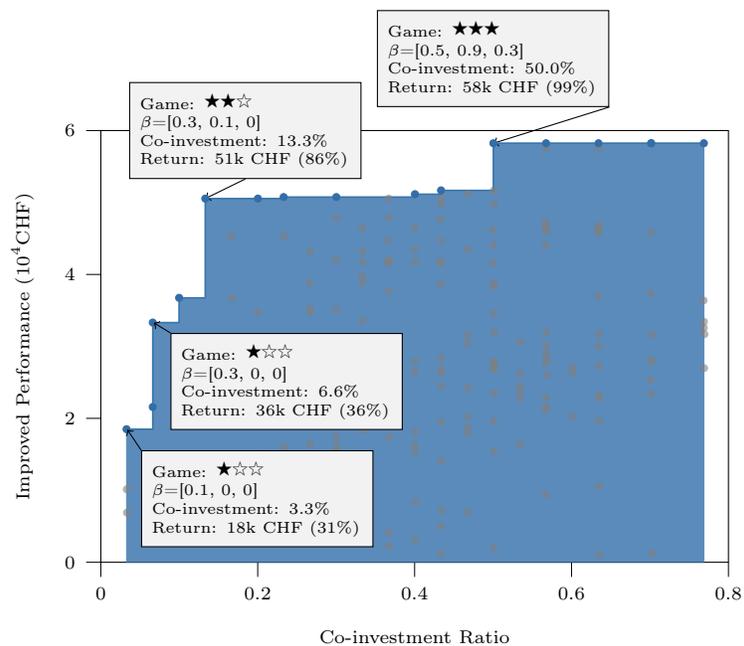}
    \caption{Co-investment ratio and improved performance.}
    \label{fig:co_ratio}
\end{figure}

\subsubsection{Heterogeneity (diversity) as an opportunity}
We next explore operator heterogeneity, focusing on differences in regional budget and intracity travel demand. 
In these experiments, the total system budget and aggregate travel demand remain fixed, but we vary Region 2's budget allocation and intracity demand (see parameter settings in \cref{tab:hete} in the Appendix).
The results show that heterogeneity can significantly amplify the benefits of the co-investment mechanism (\cref{fig:profitAB}).
The greatest improvements occur when Region 2 has a larger budget but lower intracity travel demand than Region 1.
In such cases, Region 2 can strategically fund infrastructure in Region 1 that strengthens both local service and interregional connectivity.
Region 1, in turn, leverages these improvements to expand its own service capacity despite limited resources.
While heterogeneity is often viewed as a barrier to coordination, these findings suggest that, under efficient co-investment mechanisms, it can become a structural advantage, unlocking system-level performance gains that homogeneous systems cannot achieve~\cite{lin2025heterogeneous}.
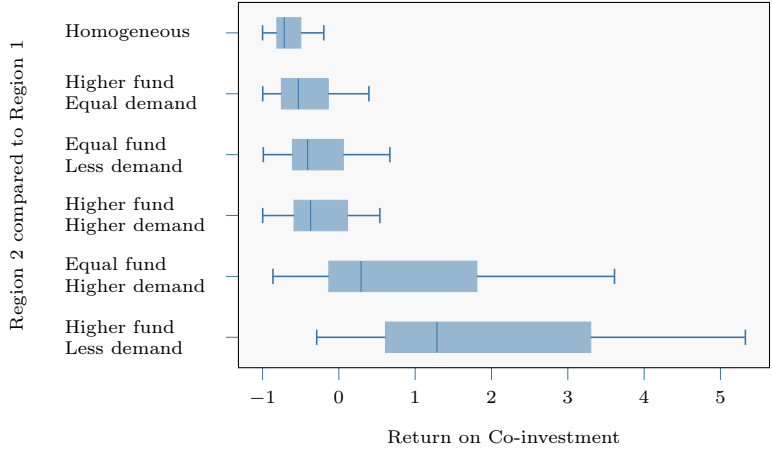
\begin{figure}[!htbp]
    \centering
    % This file was created with tikzplotlib v0.10.1.
\begin{tikzpicture}[scale=1.1]

\definecolor{darkgray176}{RGB}{176,176,176}
\definecolor{lightblue}{RGB}{173,216,230}
\definecolor{steelblueline}{RGB}{55,118,171}%{RGB}{50,110,170}
\definecolor{steelblueface}{RGB}{55,118,171}
\begin{axis}[
    width=11cm,  % Set the width of the axis
    height=7.5cm, % Set the height of the axis
    tick align=outside,
    tick pos=left,
    x grid style={darkgray176},
    xlabel={Return on Co-investment},
    xmin=-1.3161825, xmax=5.64297535714286,
    xtick style={color=steelblueline},
    y grid style={darkgray176},
    ylabel={Region 2 compared to Region 1},
    ymin=0.5, ymax=6.5,
    ytick style={color=steelblueline},
    ytick={1,2,3,4,5,6},
    yticklabels={
      {\parbox{2cm}{Higher fund \\ Less demand}},
      {\parbox{2cm}{Equal fund \\ Higher demand}},
      {\parbox{2cm}{Higher fund \\ Higher demand}},
      {\parbox{2cm}{Equal fund\\ Less demand}},
      {\parbox{2cm}{Higher fund\\Equal demand}},
      \parbox{2cm}{Homogeneous}
    },
    font=\scriptsize,
    axis background/.style={fill=lightgray,opacity=0.1},
    ]
\path [draw=steelblueline, fill=steelblueline,opacity=0.5]
(axis cs:0.61103373015873,0.75)
--(axis cs:0.61103373015873,1.25)
--(axis cs:3.30059583333333,1.25)
--(axis cs:3.30059583333333,0.75)
--(axis cs:0.61103373015873,0.75)
--cycle;
\addplot [semithick,steelblueline]
    table {%
    0.61103373015873 1
    -0.290611111111111 1
    };
\addplot [semithick, steelblueline]
    table {%
    3.30059583333333 1
    5.32665 1
    };
\addplot [semithick, steelblueline]
    table {%
    -0.290611111111111 0.875
    -0.290611111111111 1.125
    };
\addplot [semithick,steelblueline]
    table {%
    5.32665 0.875
    5.32665 1.125
    };
\path [draw=steelblueline, fill=steelblueline,,opacity=0.5]
    (axis cs:-0.131619642857143,1.75)
    --(axis cs:-0.131619642857143,2.25)
    --(axis cs:1.80814583333333,2.25)
    --(axis cs:1.80814583333333,1.75)
    --(axis cs:-0.131619642857143,1.75)
    --cycle;
\addplot [semithick, steelblueline]
    table {%
    -0.131619642857143 2
    -0.863633333333333 2
    };
\addplot [semithick, steelblueline]
    table {%
    1.80814583333333 2
    3.61175 2
    };
\addplot [semithick, steelblueline]
table {%
-0.863633333333333 1.875
-0.863633333333333 2.125
};
\addplot [semithick, steelblueline]
table {%
3.61175 1.875
3.61175 2.125
};
\path [draw=steelblueline, fill=steelblueline,opacity=0.5]
(axis cs:-0.587049206349206,2.75)
--(axis cs:-0.587049206349206,3.25)
--(axis cs:0.113916666666667,3.25)
--(axis cs:0.113916666666667,2.75)
--(axis cs:-0.587049206349206,2.75)
--cycle;
\addplot [semithick, steelblueline]
table {%
-0.587049206349206 3
-0.99784 3
};
\addplot [semithick, steelblueline]
table {%
0.113916666666667 3
0.536733333333333 3
};
\addplot [semithick, steelblueline]
table {%
-0.99784 2.875
-0.99784 3.125
};
\addplot [semithick, steelblueline]
table {%
0.536733333333333 2.875
0.536733333333333 3.125
};
\path [draw=steelblueline, fill=steelblueline,opacity=0.5]
(axis cs:-0.608140476190476,3.75)
--(axis cs:-0.608140476190476,4.25)
--(axis cs:0.058425,4.25)
--(axis cs:0.058425,3.75)
--(axis cs:-0.608140476190476,3.75)
--cycle;
\addplot [semithick, steelblueline]
table {%
-0.608140476190476 4
-0.9898 4
};
\addplot [semithick, steelblueline]
table {%
0.058425 4
0.6684 4
};
\addplot [semithick, steelblueline]
table {%
-0.9898 3.875
-0.9898 4.125
};
\addplot [semithick, steelblueline]
table {%
0.6684 3.875
0.6684 4.125
};
\path [draw=steelblueline, fill=steelblueline,opacity=0.5]
(axis cs:-0.7529375,4.75)
--(axis cs:-0.7529375,5.25)
--(axis cs:-0.13871,5.25)
--(axis cs:-0.13871,4.75)
--(axis cs:-0.7529375,4.75)
--cycle;
\addplot [semithick, steelblueline]
table {%
-0.7529375 5
-0.997857142857143 5
};
\addplot [semithick, steelblueline]
table {%
-0.13871 5
0.393316666666667 5
};
\addplot [semithick, steelblueline]
table {%
-0.997857142857143 4.875
-0.997857142857143 5.125
};
\addplot [semithick, steelblueline]
table {%
0.393316666666667 4.875
0.393316666666667 5.125
};
\path [draw=steelblueline, fill=steelblueline,opacity=0.5]
(axis cs:-0.812538888888889,5.75)
--(axis cs:-0.812538888888889,6.25)
--(axis cs:-0.498966666666667,6.25)
--(axis cs:-0.498966666666667,5.75)
--(axis cs:-0.812538888888889,5.75)
--cycle;
\addplot [semithick, steelblueline]
table {%
-0.812538888888889 6
-0.999857142857143 6
};
\addplot [semithick, steelblueline]
table {%
-0.498966666666667 6
-0.19705 6
};
\addplot [semithick, steelblueline]
table {%
-0.999857142857143 5.875
-0.999857142857143 6.125
};
\addplot [semithick, steelblueline]
table {%
-0.19705 5.875
-0.19705 6.125
};
\addplot [steelblueline]
table {%
1.285245 0.75
1.285245 1.25
};
\addplot [steelblueline]
table {%
0.291144285714286 1.75
0.291144285714286 2.25
};
\addplot [steelblueline]
table {%
-0.37039 2.75
-0.37039 3.25
};
\addplot [steelblueline]
table {%
-0.409974444444444 3.75
-0.409974444444444 4.25
};
\addplot [steelblueline]
table {%
-0.530703571428571 4.75
-0.530703571428571 5.25
};
\addplot [steelblueline]
table {%
-0.71633 5.75
-0.71633 6.25
};
\end{axis}

\end{tikzpicture}
    \caption{System improvement for heterogeneous regions.}
    \label{fig:profitAB}
\end{figure}

%Furthermore, our results indicate that operator heterogeneity provides great opportunities to enhance system performance through the proposed co-investment mechanism.
%Operator heterogeneity can be reflected in many forms, 
%in this work, we focus on the regional budget and intra-city travel demand.
%We analyze a network design case involving heterogeneous regional operators, while keeping the total system budget and aggregate travel demand constant.
%Specifically, we vary the budget allocation and intra-city travel demand in Region 2 (see parameter settings in \cref{tab:hete}, Appendix \ref{appendix:para}). 
%The figure illustrates the resulting range of return on co-investment of the system under the proposed mechanism.
%We found that consistently leads to higher system-level returns in heterogeneous settings
%The greatest improvement is observed when Region 2 has a higher budget and lower intra-city travel demand relative to Region 1.
%Through the co-investment mechanism, Region 2 strategically supports infrastructure that not only enhances Region 1’s network but also improves the overall efficiency of interregional mobility.
%Region 1 with limited resource benefits from Region 2’s resources and leverage both intraregional and interregional service capacity.  
%Heterogeneity is often seen as a challenge to coordination, but under the efficient mechanisms, it could be a structural advantage for system improvement. 
%From a policy perspective, this finding highlights 

\subsection{Zurich-Winterthur Network}\label{sec:Zurich-Winterthur Network}
The Swiss case study builds on the insight that interregional diversity can yield substantial system-wide gains.
In practice, regions often differ in size, resources, and demand patterns, characteristics that can be leveraged through co-investment strategies.
For this analysis, we extracted the \gls{acr:pt} network topologies of Zurich and Winterthur from OpenStreetMap~\cite{openstreetmap}.
Zurich's network contains 53 nodes and 66 edges; Winterthur's, 29 nodes and 34 edges.
Travel demand data is derived from a one-day transportation simulation calibrated using population data from the Swiss Federal Statistical Office~\cite{sbfsPortal}.
Based on demographic and service statistics~\cite{wikipedia_winterthur,wikipedia_zurich}, Zurich has roughly three times Winterthur’s population, and its \gls{acr:pt} system carries about ten times as many passengers per day.
Consistent with this, we assume Zurich's infrastructure budget is ten times that of Winterthur.
The redesign is implemented over a three-year horizon using a two-stage network design process.
Model parameters are provided in \cref{tab:Parameters} in the Appendix.

\subsubsection{Co-investment outcomes}
Among the many co-investment configurations tested, we focus on two representative ones (\cref{fig:zurciwinti_two}). 
\begin{figure}[!htbp]
    \centering
    \definecolor{steelblue}{RGB}{55,118,171}%{RGB}{50,110,170}
\begin{tikzpicture}[scale=1.1]
\begin{polaraxis}[
    xmin=90,
    xmax=450,
    width=6cm,
    height=6cm,
    xtick={90,162,234,306,378},
    xticklabels={
        Overall Improvement (+),
        Emission Reduction (-),
        Co-investment Ratio,
        Revenue (+),
        Customer Cost (-),
    },
    xticklabel style={
        font=\small,
        text width=3.7cm,
        align=center,
        inner sep=0pt,
        yshift=4pt 
    },
    ytick={0.2,0.4,0.6,0.8,1.0},
    yticklabels={0.2, 0.4, 0.6, 0.8, 1.0},
    yticklabel style={
    font=\small,}    ,
    % polar,
    every tick/.style={black},
    axis line style={black},
    cycle list name=color list,
    legend style={
    at={(0.5,1.25)},
    legend columns=-1,
    anchor=south,
    font=\small}
]

% Social Optimal
\addplot+[mark=*, mark options={fill=steelblue}, thick, steelblue, fill=steelblue!30, fill opacity=0.3]
coordinates {
    (90,0.57439)
    (162,0.40805)
    (234,0.75)
    (306,0.1834)
    (18,0.203219)
    (90,0.57439)
};
\node at (axis cs:90,0.57439)   
[font=\small, 
    fill=gray!10,
    rounded corners=1pt,
    anchor=south west] {0.57};

\node at (axis cs:234,0.75)   
[font=\small, 
    fill=gray!10,
    rounded corners=1pt,
    anchor=north west] {0.75};
% High Return
\addplot+[mark=square*, mark options={fill=orange}, thick, orange, fill=orange!30, fill opacity=0.3]
coordinates {
    (90,0.096000)
    (162,0.87876)
    (234,0.0833)
    (306,0.23154)
    (18,0.03396)
    (90,0.096000)
};
% Labels for values > 0.5
\node at (axis cs:152,0.65)   
[font=\small, 
    fill=gray!10,
    rounded corners=1pt,
    anchor=south east] {0.87};

\legend{Highest-Return Solution, Maximum-Efficiency Solution}

\end{polaraxis}
\end{tikzpicture}
    \caption{Two representative outcomes for Zurich and Winterthur.}
    \label{fig:zurciwinti_two}
\end{figure}
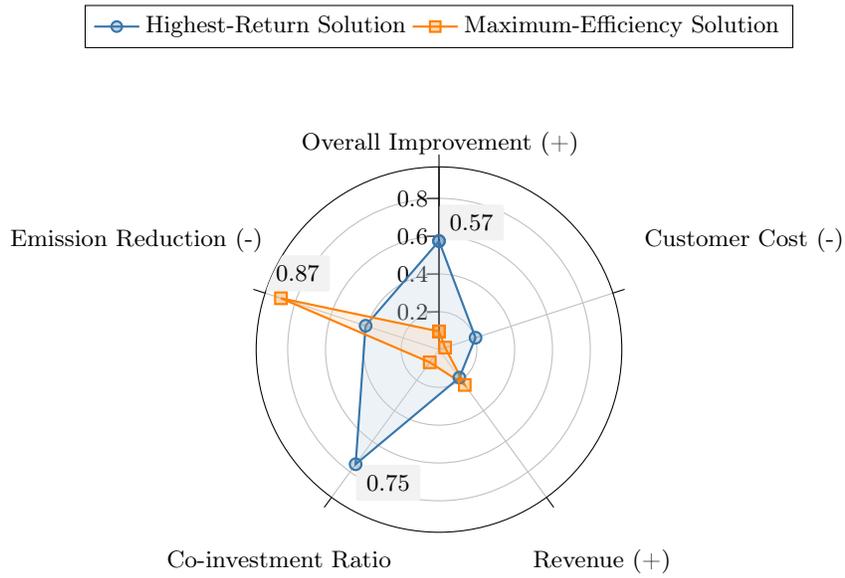
\paragraph{Highest-return solution}
Co-investment ratios of 0.25 in the first year and 1.0 in both subsequent years deliver a 57\% improvement in total performance, a 41\% reduction in emissions, and a 20\% decrease in travel costs, while \gls{acr:pt} profits increase by 18\%.
Over half of all trips are served by \gls{acr:pt} under this design, representing a 600\% increase in ridership relative to the non-cooperative baseline.
\paragraph{Highest-efficiency solution}
A one-time co-investment of just 25\% of the first-year budget yields a 10\% improvement in system performance, an 87\% reduction in emissions, a 3\% decrease in travel costs, and a 23\% increase in \gls{acr:pt} revenue.
\gls{acr:pt} demand rises by 110\%.
This confirms the pattern observed in the Sioux Falls case (\cref{sec:System Improvement from Co-investment}): small, well-timed investments can yield disproportionately large efficiency gains.

\subsubsection{Strategic payoff-sharing}
In real-world negotiations, operators with smaller budgets often have weaker bargaining power and may be more vulnerable to strategic exploitation (see Section~\ref{sec:Discussion on Bargaining Power and Exploration}).
To examine these effects, we evaluate how Winterthur's improvement changes with bargaining power, the presence of exploitation, and the co-investment ratio, a proxy for the depth of cooperation.
We consider four scenarios (\cref{fig:winti}):
a) symmetric bargaining power, no exploitation,
b) symmetric bargaining power, with exploitation,
c) asymmetric bargaining power, no exploitation,
d) asymmetric bargaining power, with exploitation.
In this context, asymmetric bargaining power refers to a setting in which the payoff distribution is associated with the respective investment amounts of operators (in \cref{eq:alpha}). 
Exploitation indicates that the shareable payoff is generated solely from Winterthur, which is captured in \cref{eq:epsilon_qQ} and \cref{eq:epsilon_vq} by setting~$\epsilon_\text{zurich} = 0, \epsilon_\text{winti} = 1$. 
Non-exploitation allows payoff generation in both regions, with~$\epsilon_\text{zurich} = 1, \epsilon_\text{winti} = 1$.
We assume both regions use the same co-investment ratio.
\begin{figure*}[!htbp]
    \centering
    \begin{subfigure}[t]{0.495\textwidth}
        \centering
        \includegraphics[width=0.9\linewidth]{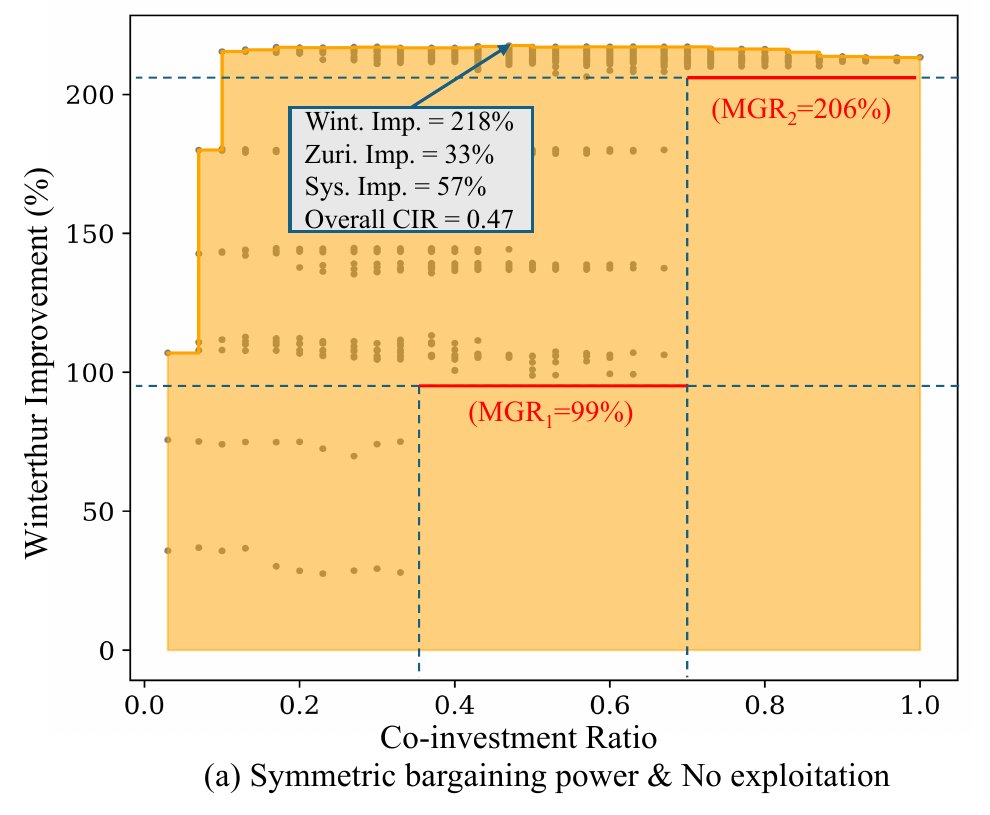}
        \vspace{-1em}
        %\caption{Symmetric bargaining power \& no exploitation}
        \label{fig:winti_case_a}
    \end{subfigure}
    \begin{subfigure}[t]{0.495\textwidth}
        \centering
        \includegraphics[width=0.9\linewidth]{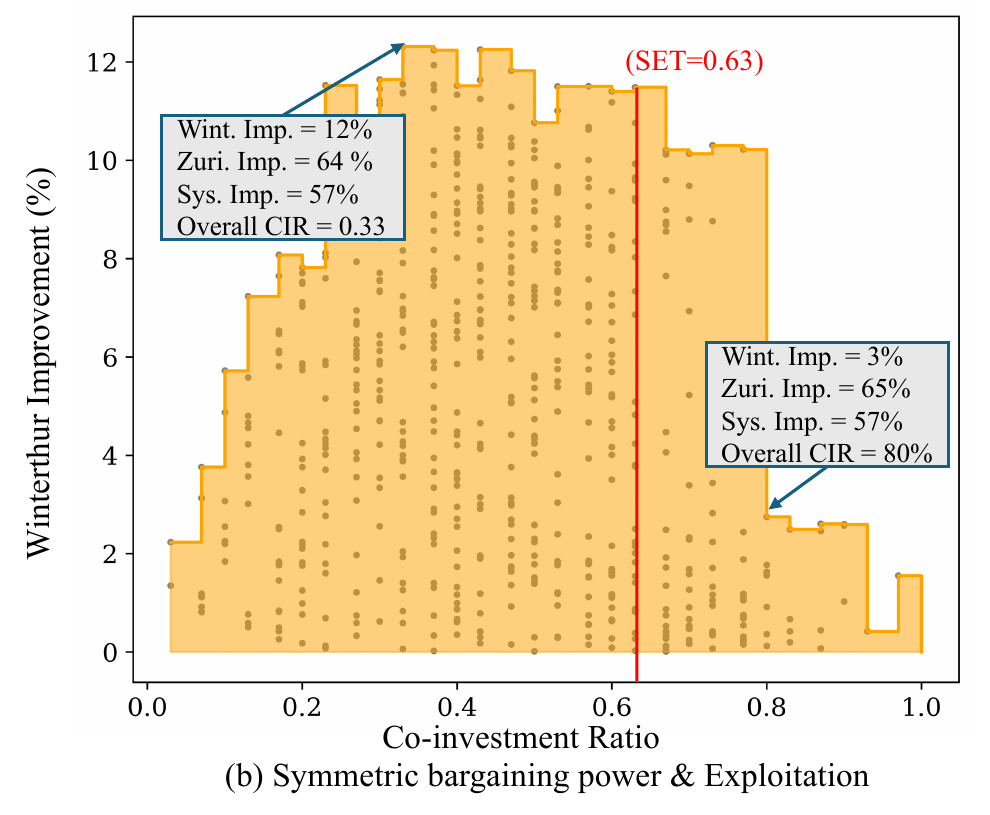}
        \vspace{-1em}
        %\caption{Symmetric bargaining power \& exploitation}
        \label{fig:winti_case_b}
    \end{subfigure}
    %\vspace{0.8em}
    \begin{subfigure}[t]{0.49\textwidth}
        \centering
        \includegraphics[width=0.9\linewidth]{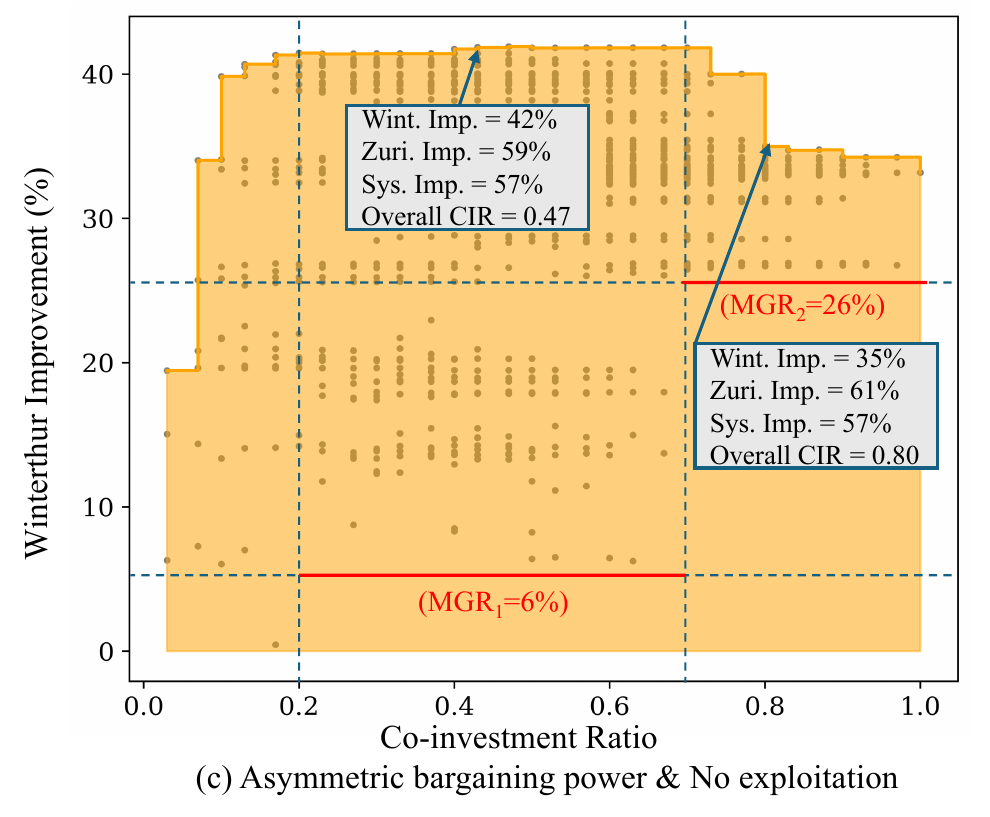}
        \vspace{-1em}
        %\caption{Asymmetric bargaining power \& no exploitation}
        \label{fig:winti_case_c}
    \end{subfigure}
    \begin{subfigure}[t]{0.49\textwidth}
        \centering
        \includegraphics[width=0.9\linewidth]{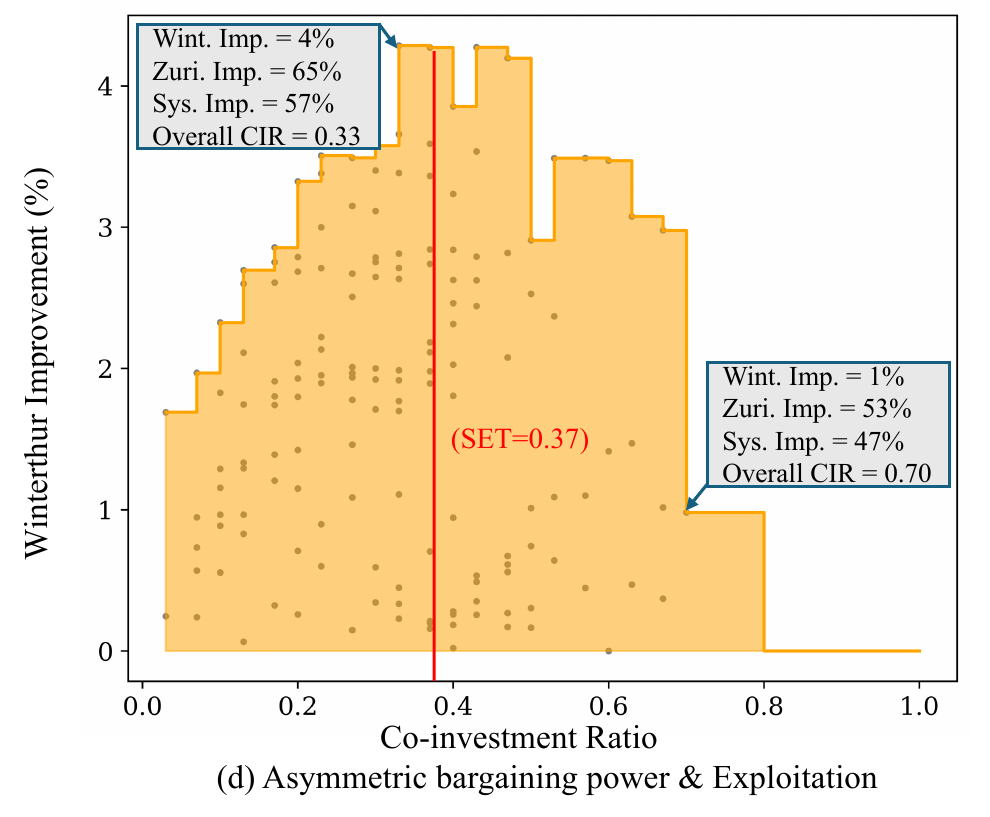}
        \vspace{-1em}
        %\caption{Asymmetric bargaining power \& exploitation}
        \label{fig:winti_case_d}
    \end{subfigure}
    \caption{Co-investment ratio and improved performance for Winterthur under four payoff-sharing scenarios.}
    \label{fig:winti}
\end{figure*}
   
\paragraph{Bargaining power effects}
Winterthur achieves its largest gains under symmetric bargaining with no exploitation (\cref{fig:winti}a), where even low co-investment ratios deliver strong returns, mirroring the pattern in \cref{fig:co_ratio}.
In contrast, when bargaining power is asymmetric, deeper cooperation can reduce Winterthur's benefits.
For instance, in \cref{fig:winti}c, increasing the co-investment ratio from 0.47 to 0.80 raises Zurich's benefit slightly (59\% to 61\%) but lowers Winterthur's from 42\% to 35\%.
This effect is absent under symmetric bargaining.

\paragraph{Strategic exploitation effects}
Without exploitation, a minimum guaranteed return (MGR) emerges for Winterthur.
In \cref{fig:winti}a, Winterthur is guaranteed a 99\% improvement once the co-investment ratio exceeds 0.37, rising to 206\% for ratios above 0.7.
Even under asymmetric bargaining (\cref{fig:winti}c), MGRs of 6\% and 26\% are observed.
However, in \cref{fig:winti}b and \cref{fig:winti}d (with exploitation), no MGR exists; Winterthur's benefits can fall to zero.

The SET marks the point beyond which higher co-investment harms the weaker operator.
In \cref{fig:winti}b, this occurs at a ratio of 0.63;
in \cref{fig:winti}d, at 0.37.
Beyond these points, Winterthur's returns decline sharply: in \cref{fig:winti}b, increasing the ratio from 0.33 to 0.8 raises Zurich's benefit from 64\% to 65\%, but slashes Winterthur's from 12\% to 3\%.
In \cref{fig:winti}d, the drop is from 4\% to 1\% when the ratio increases from 0.33 to 0.7.
%Above this threshold, continued cooperation offers Winterthur no advantage over the non-cooperative baseline, eliminating its incentive to participate.
\mh{
The SET can be interpreted as an incentive signal: beyond this point, the weaker operator's return declines, so it has no reason to commit to a higher co-investment ratio. This suggests that the strategic choice of the co-investment ratio should be formalized, forming a meta-game on top of the proposed two-stage model.}

\mh{
\begin{remark}[Sensitivity of the SET to the budget ratio]
The SET can be sensitive to the infrastructure budget ratio, although the dependence is generally indirect and need not be monotone. In our framework, the budget ratio affects the payoff-sharing outcome through two channels. First, in asymmetric bargaining, it changes the bargaining power $\alpha_{\mathrm{wint}}=\beta_{\mathrm{wint}} B_{\mathrm{wint}}/(\beta_{\mathrm{wint}} B_{\mathrm{wint}}+ \beta_{\mathrm{zuri}} B_{\mathrm{zuri}})$. Thus, when the same co-investment ratio is used, a smaller $B_{\mathrm{winti}}/B_{\mathrm{zuri}}$ reduces Winterthur's bargaining weight and can make it recover a smaller share of the surplus under exploitation. Second, the budget ratio also affects the disagreement point $\varphi$, which is determined by the fully non-cooperative outcome with full budgets. A larger budget may improve an operator's outside option and its bargaining position. Since these two effects interact with the resulting network designs, the overall impact of the budget ratio on the SET is case-dependent.
\end{remark}}

% \begin{remark}[On the choice of the co-investment fraction]
% \label{rem:beta_endogeneity}
% We treat the co-investment fraction $\beta_i$ as fixed, while our results suggest which values a rational operator would choose. \Cref{fig:co_ratio} shows that some co-investment ratios are more efficient than others, and \cref{fig:winti} shows that beyond a certain point, committing more budget can lower an operator's own return. 

% An interesting direction would let both operators choose $\beta_i$ simultaneously, each anticipating the cooperative gains. 
% This turns the network design problem into a multi-stage game, whose stages would determine the cooperation budget, the joint projects to build, the benefit allocation, as well as the local planning decisions. 
% Several considerations make this nontrivial. The disagreement point becomes endogenous, since what an operator falls back on when negotiation fails now depends on the budget it has already committed and on any alternative cooperation strategy available to it. 
% The operators' choices are coupled through their bargaining shares. 
% Finally, one should consider whether an operator has the incentive to commit to its announced $\beta_i$ truthfully. We leave this analysis to future work.
% \end{remark}
% }

\mh{
\section{Discussion: Scope and Limitations}
\label{sec:limitations}
The purpose of this section is not only to acknowledge the assumptions and limitations, but also to clarify precisely which parts of the framework already extend beyond them and which parts would require new analysis.

\textbf{Number of operators.} 
Although our experiments consider two operators, the general game (Definition~1), the co-investment problem, and the
payoff-sharing result (Proposition~5) are stated for an arbitrary operator set $\mathcal{I}$. The Nash bargaining solution extends directly to a grand coalition of $N>2$ operators, as its existence and uniqueness require only a strictly positive surplus over a convex feasible set, independent of $|\mathcal{I}|$. What $N>2$ adds is coalition formation: sub-coalitions may prefer to cooperate among themselves, so the grand-coalition allocation must additionally be stable against such deviations. We isolate the strategic and payoff-sharing mechanisms in the two-region setting and leave the coalition-structure analysis to future work.

\textbf{Perfect information.}
The observability of network designs and budgets is reasonable for public infrastructure (Section~2.1). Relaxing it yields a Bayesian game in which each operator optimizes its expected payoff over beliefs about the others' layout, budget, disagreement point, and bargaining power; the disagreement point and weights are replaced by their expected counterparts, and individual rationality holds in expectation rather than ex post.

\textbf{Congestion.} The congestion-free model suits public transport on dedicated infrastructure (rail, segregated bus) and is least accurate for modes sharing road capacity. It is also what preserves our guarantees: with congestion the lower level becomes a user-equilibrium problem and the design problem a multi-leader--multi-follower MPEC, forfeiting convexity (Lemma~1), global optimality, and the existence results of Section~4~\cite{he2025hierarchical}.

\textbf{Route choice.}
Based on the two-route model (a PT-prioritized route and an alternative), we derive the conditions in Lemma~1.
Multi-route or nested-logit specifications are compatible with the framework but require re-verifying the convexity conditions~(16)--(17).

\textbf{Continuous relaxation and integer computation}
Our equilibrium existence guarantee (Proposition~\ref{prop:existence}) is limited to the continuous relaxation; the integer game does not inherently guarantee pure-strategy equilibria. 
For our case studies, we computed equilibria directly over integer action sets via IBR, allowing a 0.5\% optimality gap. 
Additionally, while IBR carries no a priori convergence guarantee for discrete general-sum games, empirical convergence was achieved in every instance tested (Table~\ref{tab:ibr-conv}). 
Extending these theoretical guarantees to the exact integer formulation remains open for future work.
}
\section{Conclusion}\label{sec:conclusion}
In this work, we proposed a game-theoretic framework for network design in settings where multiple self-interested operators make strategic decisions.
The framework formalizes the network design game for analyzing non-cooperative interactions, and subsequently introduces a co-investment and payoff-sharing mechanism to foster mutually beneficial cooperation in competitive environments.

The approach was first validated on the Sioux Falls benchmark network and then applied to the real-world \gls{acr:pt} systems of Zurich and Winterthur, Switzerland.
Across both cases, the proposed mechanism consistently improved network designs for the benefit of both operators and users.
A notable insight is that even modest, well-timed co-investments can deliver substantial gains across environmental, social, and economic dimensions.
Furthermore, regional heterogeneity, often seen as a coordination challenge, emerged as a structural advantage when leveraged through cooperative design.
In the Zurich-Winterthur study, we also examined bargaining power and strategic exploitation, finding that these factors strongly shape both the distribution of benefits and the incentives to deepen cooperation.
\mh{Based on our results on MGR and SET, we recommend that coordinating institutions support weaker participants by subsidizing them to meet the investment threshold for minimum returns. Moreover, institutions can use expected MGR values to identify exploitation and mandate a full sharing of cooperative gains.
These measures ensure that cooperation remains beneficial for smaller municipalities, thereby sustaining their willingness to participate.}

Looking ahead, this framework can be extended to multimodal transit and broader stakeholder networks to enrich sustainable policy analysis. 
Furthermore, the model generalizes to other interdependent infrastructures, such as energy and communication grids. 
Finally, accounting for dynamic factors like network evolution, population migration, and implementation uncertainty will improve its viability for long-term strategic planning.

% Looking ahead, the framework can be extended in several directions.
% Incorporating additional real-world transportation systems and a wider set of public and private stakeholders would enable richer policy analysis.
% Applying the mechanisms to multimodal settings, integrating rail, bus, ridesharing, and other services, can produce practical strategies for sustainable and coordinated mobility.
% Beyond transportation, the model can be adapted to other networked infrastructures such as energy grids and communication systems, where local optimization and global performance are inherently linked.
% Finally, incorporating network evolution, population migration, and uncertainty in agreement implementation would improve the framework's applicability to long-term, real-world planning.
%

\appendix

\section{Alternative Demand Model}
\label{appendix:traffic assignment}

In this work, we assume that congestion does not affect route choices. 
To incorporate congestion effects, we extend the demand model to capture strategic traveler interactions.
Following Wardrop’s first principle~\cite{wardrop1952road}, the UE traffic assignment assumes that travelers selfishly choose their routes, achieving equilibrium when no traveler can improve their outcome by unilaterally changing their route. 
In line with UE traffic assignment, we formulate the following optimization problem for user-level modeling for the regional network design problem:
\begin{subequations} \label{opt:ue}
\begin{align}
    \min_{y_e} &\sum_{e \in \edges}\int_{0}^{w} g_e \tup{y_e,x_e} \,dw,\\
    \text{s.t.} \quad
    &\sum_{k \in \mathcal{K}} f_m^k=\alpha_{m}\\
    &f_{m}^{k} \geq 0\\
    &y_e=\sum_{m \in \mathcal{M}} \sum_{k \in \mathcal{K}} f_{m}^{k}  \mathds{1}_{\edges^k_m\tup{e}}\\
    &0 \leq y_e \leq c_e, \quad \forall e \in \edges_P, \label{ue:cap}
    %& c_e<x_e \Omega, \forall e \in \edges_R,
\end{align}
\end{subequations}
where the function $g_e: \nonnegativeintegernumbers \rightarrow \nonnegativenumbers$ maps the edge flow to the generalized travel cost on edge $e$ in the multimodal transportation network (see \cref{eq:g}). 
Specifically, the formulation accounts for both congestion of road traffic and the availability of the \gls{acr:pt} service. 
The generalized travel cost includes a monetary valuation of time and a distance-based transportation fee.
The cost calculation differs depending on the transport mode. 
For road traffic on the alternative-mode layer, the Bureau of Public Roads (BPR) function is used to estimate travel time \cite{united1964traffic}. 
In contrast, \gls{acr:pt} services are assumed to operate on dedicated infrastructure, which is not affected by road congestion.
When the \gls{acr:pt} connection is unavailable (i.e., $ x_e = 0$), the travel cost on edge $e$ is set to a large positive constant $\Omega$.
\begin{align}
    g_e \tup{y_e,x_e} =
    \begin{cases}
    \gamma_\mathrm{vot} \hat{t}_{e}  \tup{ 1+ a \tup{\frac{y_e}{c_e}}^{b}}+l_e \gamma_1^A
    , & \text{if } e \in \edges_A,\\
    x_el_e\tup{\frac{\gamma_\mathrm{vot}}{v_P}+\gamma_1^P}+\tup{1-x_e}\Omega
    , &\text{otherwise}.
    \end{cases}
    \label{eq:g}
\end{align}
where $a$ and $b$ are the parameters of the BPR function, and $\hat{t}_{e}$ is the free flow time of edge $e$.
Then, the multi-regional network design problem can be structured as a Stackelberg game, with operators as leaders and travelers as followers. 
This captures the hierarchical nature of decision-making, involving interactions among travelers at the lower level and network designers at the upper level. 
The optimization-based demand model in \cref{opt:ue} is an alternative component in the general \gls{acr:ndg} framework (\cref{sec:Network Design Game}) and provides a basis for \gls{acr:ndg} with hierarchical structures.

\section{Model parameters and Verification of Convexity Conditions}
\label{appendix:para}
\subsection{Model parameters}
\Cref{tab:hete} shows the parameters used in the heterogeneous-region scenarios in the Sioux Falls case study, and \cref{tab:Parameters} presents model parameters for the  Zurich-Winterthur case.
% The type of trips $\theta_m \in \Theta = \{\theta^1_\mathrm{intra}, \theta^2_\mathrm{intra},\theta^1_\mathrm{inter}, \theta^2_\mathrm{inter}\}$.
A travel request can be classified as \emph{intra-regional}, if both its origin and destination belong to the same region, i.e, $o_m, d_m \in \nodes_i$. 
It is classified as \emph{inter-regional} if the origin and destination are in different regions, i.e., $o_m \in \nodes_i, d_m \in \nodes_j$, with $i,j \in \operatorSet, i \ne j$. 
We use $|R(\Theta_1^{intra})|$ and $|R(\Theta_2^{intra})|$ to denote the number of intra-regional trips originating in Region~1 and Region~2, respectively.

\begin{table}[tb]
    \centering
    \caption{Scenario parameters for region heterogeneity.}
    \small
    \begin{tabular}{lcc}
    \toprule
     Scenarios & $B_1:B_2$ & $|R(\Theta_1^{intra})|:|R(\Theta_2^{intra})|$   \\ \midrule
     Homogeneous & 1:1 & 1:1  \\ 
     Higher fund, Equal pop& 3:2 & 1:1\\
     Equal fund, Less pop & 1:1 & 2:3\\
     Higher fund, Higher pop& 3:2 & 3:2\\
     Equal fund, Higher pop& 1:1 & 3:2\\
     High fund,  Less pop& 3:2 & 2:3\\
    \bottomrule
    \end{tabular}
    \label{tab:hete}
\end{table}

\begin{table}[tb]
    \centering
    \caption{Model parameters.}
    \small
    \begin{tabular}{lllll}
    \toprule
        Parameters& Description & Value  & Unit & Ref. \\ \midrule
        \multicolumn{5}{l}{\textbf{Network design}} \\
        $B_\mathrm{zuri}$ & Budget for Zurich & $16 \times 10^4$ &  CHF/day & - \\ 
        $B_\mathrm{winti}$ & Budget for Winterthur& $1.6 \times 10^4$ &  CHF/day & - \\ 
        % $\beta^t_i$ & \multicolumn{4}{l}{Regional co-invest ratio \{0, 0.1, 0.3, 0.5, 0.7, 0.9\}}  \\
        $c^{b}$ & Base cost & 91 &  CHF/day/km & \cite{flyvbjerg2008comparison} \\ 
        $c^{k}$ & Capacity cost & 84 & CHF/day/km & \cite{bus_speed} \\ 
        $s_\mathrm{max}$ & Maximum frequency & 20 & veh/h& ~\\ 
        % $a,b$ & BPR function & 4, 0.15 & -& \cite{united1964traffic}\\
        $\Omega$ & Large number& $1 \times 10^8$ & -& -\\
        \multicolumn{5}{l}{\textbf{Travel demand}} \\
        $\tau$ & Demand growth rate & 1.5 &\% & \cite{growthrate}\\
        $\gamma_\mathrm{vot}$ & Value of time & 30  & CHF/h & \cite{schmid2021value} \\ 
        %\multicolumn{5}{l}{\textbf{Mobility service}} \\
        \textbf{Public transit} & ~ & ~ & ~ & ~ \\ 
        $\gamma_{1}^{P}$ & Service fee & 0.092  & CHF/km/pax & \cite{zvv} \\ 
        $\gamma_{2}^P$ & Emission & 0 & kg/km/pax & - \\ 
        $v_{P}$ & Speed & 50 & km/h & \cite{bus_speed} \\ 
        $\kappa$ & Capacity & 60  & seat/veh & \cite{transit_cap} \\
        \multicolumn{5}{l}{\textbf{Alternative mode}} \\
        $\gamma_{1}^{A}$ & Service fee & 0.65 & CHF/km/pax & \cite{mobilityeasy2025} \\ 
        $\gamma_{2}^A$ & Emission & 0.148 & kg/km/pax & \cite{em} \\ 
         $v_{A}$ & Speed & 60 & km/h & \cite{taxispeed} \\ \bottomrule
    \end{tabular}
    \label{tab:Parameters}
\end{table}
\vspace{-0.2cm}
\mh{
\subsection{Verification of Convexity Conditions}
To ensure the theoretical guarantees of Lemma~1 apply with those model parameters, we verify that convexity Conditions~(16) and (17) hold. 
Condition~(16) requires $u_m^A - u_m^P \leq 0$ for all travel requests $m$. We confirm this numerically for both networks: the condition is strictly satisfied across all requests, with worst-case maximums of -0.611 for Sioux Falls and -0.058 for Zurich--Winterthur. 
Condition~(17) requires $\Delta_e \geq 0$ for every public transit edge $e$. The value of $\Delta_e$ ranges from 0.616 to 2.556 in Sioux Falls, and from 0.058 to 1.106 in Zurich--Winterthur. 
Because both conditions are met (summarized in Table~\ref{tab:convexity}), the operator subproblems are guaranteed to remain mixed-integer convex programs throughout the experiments.

\begin{table}[htb]
    \color{tablecolor}
    \centering
    \caption{\mh{Verification of convexity conditions.}}
    \small
    \begin{tabular}{lcc}
    \toprule
     & Sioux Falls & Zurich--Winterthur \\
    \midrule
    \multicolumn{3}{l}{\emph{Condition (16): $u_m^A - u_m^P \leq 0$}} \\
    $\max\,(u_m^A - u_m^P)$ & -0.611 & -0.058 \\
    \midrule
    \multicolumn{3}{l}{\emph{Condition (17): $\Delta_e \geq 0$}} \\
    $\min \Delta_e$ & 0.616 & 0.058 \\
    $\max \Delta_e$ & 2.556 & 1.106 \\
    \bottomrule
    \end{tabular}
    \label{tab:convexity}
\end{table}

}

\mh{
\section{Empirical Convergence and Computational Performance}

\subsection{Empirical Convergence of \gls{acr:ibr}}
\label{appendix:ibr-conv}

Pure-strategy Nash equilibria need not exist in general-sum discrete
games, and computing them is hard even when they
do~\cite{daskalakis2009complexity,chen2009settling}. We therefore make
no general convergence claim for the \gls{acr:ibr} algorithm used to compute the Stage-1 equilibrium. Instead, we treat it as an empirical equilibrium-finding procedure that yields an \emph{a posteriori} certificate of the Nash property whenever it converges.
We evaluated the convergence dynamics of IBR across 108 scenarios on two networks. The experiments varied four key parameters: the co-investment ratio ($\beta \in \{0, 0.25, 0.5, 0.75, 1\}$), the per-operator budget ($5 \times 10^4$ to $4 \times 10^6$), budget asymmetry (from 1:1 to 10:1), and inter-regional demand intensity (1x to 8x).

\begin{table}[tb]
\color{tablecolor}
\centering
\caption{\mh{Empirical convergence of IBR.}}
\label{tab:ibr-conv}
\small
\begin{tabular}{@{}lcccccc@{}}
\toprule
Case & Runs & \multicolumn{3}{c}{Iterations to equilibrium} & \multicolumn{2}{c}{$\varepsilon$ results} \\
\cmidrule(lr){3-5} \cmidrule(l){6-7}
& & $\leq 2$ & $3$--$6$ & osc. & Exact ($\varepsilon=0$) & Max $\varepsilon$ \\
\midrule
Sioux Falls        & 54  & 47 & 2 & 5 & 49 & $0.16\%$ \\
Zurich--Winterthur & 54  & 44 & 6 & 4 & 50 & $0.12\%$ \\
\midrule
Total              & 108 & 91 & 8 & 9 & 99 & - \\
\bottomrule
\end{tabular}
\end{table}
\vspace{-0.2cm}
\subsection{Computational Details and Scalability}

All models are solved using Gurobi 12.0.3 (to a 0.5\% optimality gap target) via Python 3.12.8 on the ETH Euler cluster using AMD EPYC 64-core nodes. 
%The framework's tractability relies on two structural features. First, the Iterated Best Response (IBR) decomposition avoids solving the full network game simultaneously. Second, by Lemma~1, the Stage-1 operator subproblems and the Stage-2 co-investment problem are mixed-integer convex programs (MICPs), avoiding the exponential worst-case of generic MINLPs. The Stage-2 payoff-sharing problem~(13) is convex and solves almost instantly.
%
Table~\ref{tab:compute} summarizes problem sizes and computational times for the Sioux Falls and Zurich--Winterthur networks (co-investment ratio 0.5). While branch-and-bound runtimes vary across the 10 Gurobi seeds, the resulting equilibria, IBR iteration counts, and built networks are seed-invariant.
As shown in \cref{fig:scaling}, wall-clock time grows as a low-order power of the model size ($|\mathcal{M}|\cdot|\edges|$). Runtime is $\propto(|\mathcal{M}||\edges|)^{0.92}$ for subsampled Zurich--Winterthur instances and $\propto(|\mathcal{M}||\edges|)^{1.30}$ for aggregated New York City instances ($R^2\ge0.96$). The branch-and-bound search scales efficiently: the largest tested instance (1,000 nodes, $|\mathcal{M}||\edges|\approx4\times10^{6}$) achieves a gap less than $0.5\%$ in 1{,}880s.
\begin{table}[tb]
    \color{tablecolor}
    \centering
    \caption{\mh{Computational performance.}}
    \small
    \begin{tabular}{lcc}
    \toprule
    & Sioux Falls & Zurich--Winterthur \\ \midrule
    \textbf{Network Size} & & \\
    Nodes & 24 & 82 \\
    PT edges & 76 & 100 \\ \midrule
    \textbf{Optimization Problem Size} & & \\
    Build variables (Binary) $|\{d_e\}|$ & 76 & 100 \\
    Capacity variables (Continuous) $|\{s_e\}|$ & 76 & 100 \\
    Flow variables (Continuous) & 3{,}032 & 33{,}007 \\
    Total variables & 3{,}184 & 33{,}207 \\
    Constraints & 6{,}979 & 15{,}465 \\
    \midrule
    \textbf{Computational Time} & & \\
    IBR iterations to converge & 2 & 3–4 \\
    Time per IBR iteration (s) & 12.4 [11.9, 13.2] & 77.9 [61.6, 131.2]\\ 
    Stage-1 time per game(s) & 24.8 [23.7, 26.3] & 202 [148, 394] \\ 
    Stage-2 co-investment time (s) & 6.0 [5.9, 6.1]	& 36.0 [24.7, 47.7] \\  
    Stage-2 Payoff-sharing time (s) & $<0.01$ & $<0.01$ \\
    Optimality gap & 0.43\% & 0.49\% \\
    \bottomrule
    \end{tabular}
    \label{tab:compute}
    \par\smallskip
    \parbox{0.72\textwidth}{\footnotesize \textit{Note:} Computational times are presented as \emph{median [min, max]} over 10 Gurobi seeds. All times reflect a single planning game.}
\end{table}

\begin{figure}[tb]
    \centering
    \color{tablecolor}
    \includegraphics[width=0.7\linewidth]{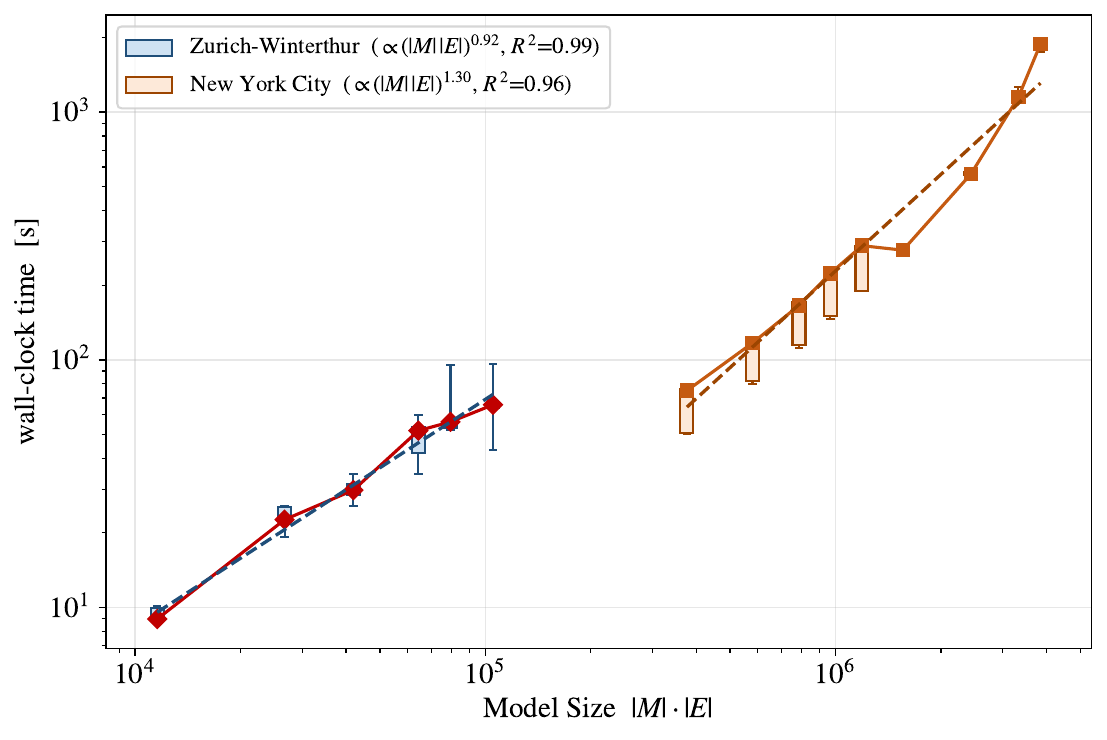}
    \caption{\mh{Wall time versus model size.}}
    \label{fig:scaling}
    
    \vspace{0.15cm}
    \begin{minipage}{0.95\linewidth} % Adjust width to match your column/preference
        \footnotesize \renewcommand{\baselinestretch}{0.8}\selectfont
        \textit{Note:} Model size is defined as $|\mathcal{M}|\cdot|\edges|$ (active OD pairs $\times$ network links). Data reflects a single fixed instance per size evaluated across five solver seeds. Boxes span the min–max range, the solid line indicates the median, and dashed lines represent power-law fits.
    \end{minipage}
\end{figure}
}

\mh{
\section{Robustness to Bargaining Power}

In this work, we assign asymmetric Nash bargaining weights based on financial contributions, which is an established practice for allocating surplus in cooperative infrastructure and energy models~\cite{han2005fair,nazari2021joint,ding2025cooperative}. 
However, real-world bargaining leverage can also stem from factors like network centrality, political influence, or an operator's infrastructure upgrade potential. 
% Eq.~(14) offers the flexibility to incorporate any of these metrics, and we conduct additional experiment to demonstrate the robustness of our core findings.
To demonstrate that our qualitative conclusions are robust to this modeling choice, we evaluate an alternative bargaining rule based on network size. 
We use the Sioux Falls network and conduct three year of planning with a discrete grid of co-investment ratios ($\beta \in \{0, 0.25, 0.5, 0.75, 1\}$), strictly enforcing the condition $\beta_2 \leq \beta_1$ to maintain Region 1's investment dominance.
Comparing \cref{fig:Bargaining power}a and b confirms that SET emerges under both rules; Region 2’s improvement consistently declines once its ratio reaches 0.75, showing a steeper drop (0.43\%) under the network-size rule. 
Furthermore, MGR also persists under the network-size rule once Region 1 stops exploiting Region 2 (\cref{fig:Bargaining power}b,c), yielding gains of 0.2\% over the 0.25–0.50 range and 0.51\% over the 0.58–0.75 range. 

\begin{figure}[tbp]
    \centering
    \begin{subfigure}[b]{0.48\textwidth}
        \centering
        \includegraphics[width=\textwidth]{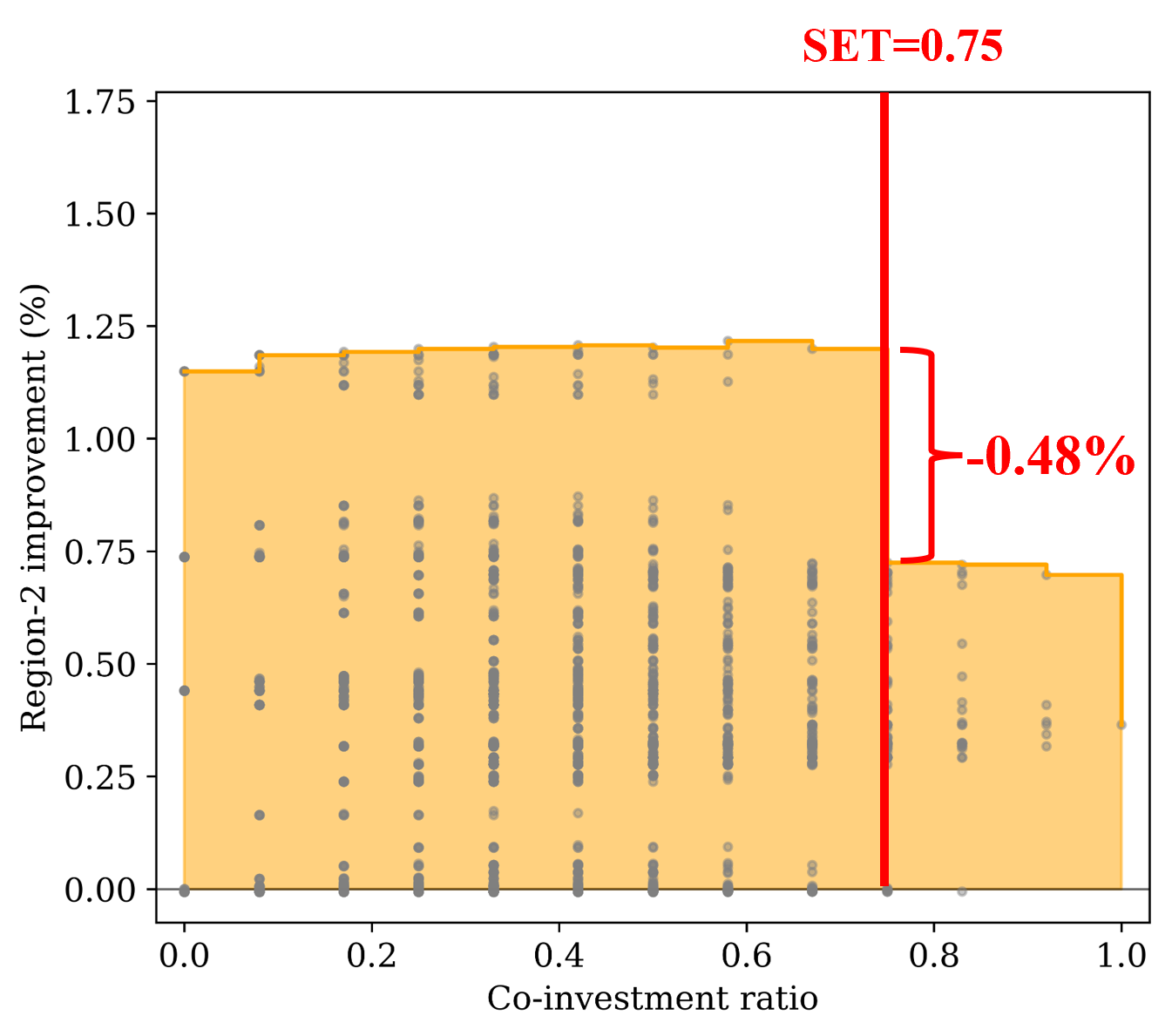}
        \caption{Contribution, Exploitation}
        \label{fig:sub1}
    \end{subfigure}
    \hfill
    \begin{subfigure}[b]{0.48\textwidth}
        \centering
        \includegraphics[width=\textwidth]{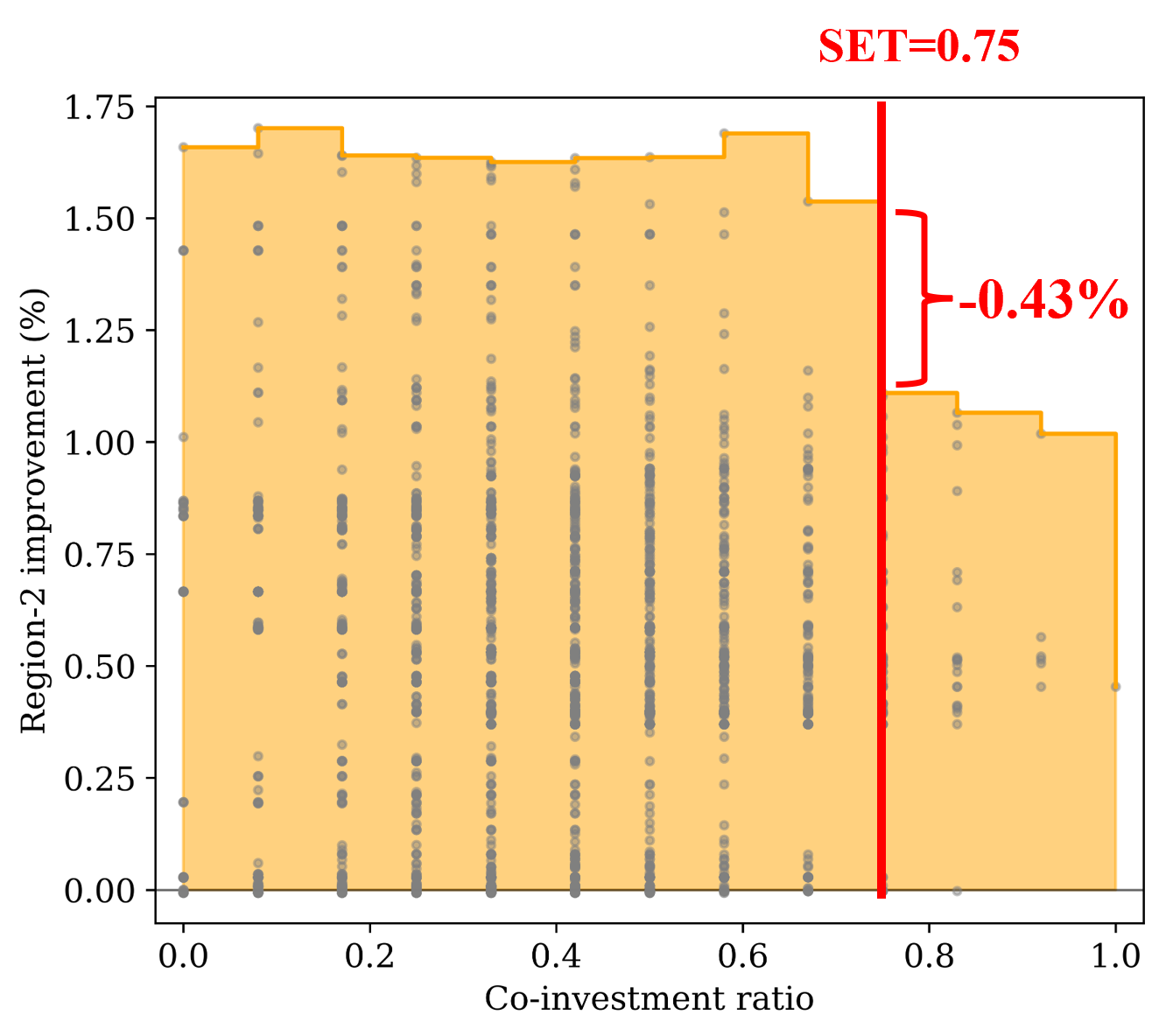}
        \caption{Network size, Exploitation}
        \label{fig:sub2}
    \end{subfigure}

    \vspace{0.4cm}
    \begin{subfigure}[b]{0.48\textwidth}
        \centering
        \includegraphics[width=\textwidth]{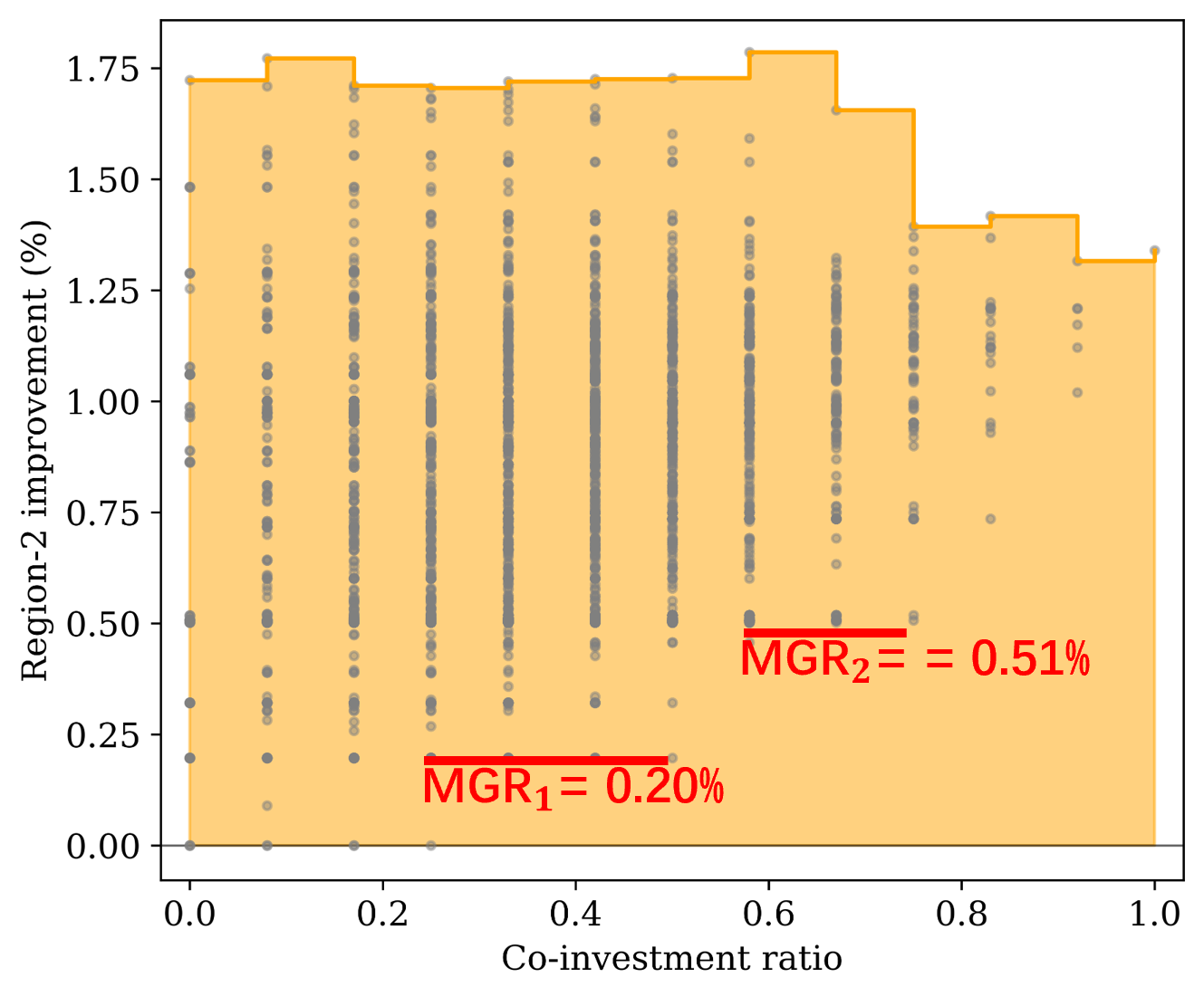}
        \caption{Network size, No exploitation}
        \label{fig:sub3}
    \end{subfigure}
    \caption{\mh{Distribution of performance improvements for Region 2 under varying bargaining power and exploitation scenarios.}}
    \label{fig:Bargaining power}
\end{figure}

}
\clearpage
\bibliographystyle{unsrtnat}
\bibliography{ref}

\end{document}